\pdfoutput=1
\documentclass[]{article}

\usepackage{todonotes}

\usepackage{graphicx}
\usepackage{amsfonts}
\usepackage{amsmath}
\usepackage{amssymb}
\usepackage{fancyhdr}
\usepackage{titlesec}
\usepackage{indentfirst}
\usepackage{booktabs}
\usepackage{verbatim}
\usepackage{color}
\usepackage{amsthm}
\usepackage{url}
\usepackage{cite}
\usepackage{epstopdf}
\usepackage[page,header]{appendix}
\usepackage{titletoc}

\numberwithin{equation}{section}
\newtheorem{theorem}{Theorem}[section]
\newtheorem{lemma}[theorem]{Lemma}
\newtheorem{corollary}[theorem]{Corollary}
\newtheorem{proposition}[theorem]{Proposition}
\newtheorem{definition}[theorem]{Definition}
\newtheorem{remark}[theorem]{Remark}

\newcommand{\rd}{\mathrm{d}}

\newcommand{\bR}{\mathbb{R}}

\newcommand{\Sd}{\mathbb{S}}

\newcommand{\vs}{v_{*}}

\newcommand{\am}{a_{-}}
\newcommand{\ap}{a_{+}}
\newcommand{\xie}{\xi_{e}}
\newcommand{\p}{\varphi}
\newcommand{\K}{\mathcal{K}}

\newcommand*{\im}{\mathop{}\!\mathrm{i}}
\newcommand*{\e}{\mathop{}\!\mathrm{e}}

\begin{document}

\title{Measure Valued Solution to the Spatially Homogeneous Boltzmann Equation with Inelastic Long-Range Interactions}

	\author{Kunlun Qi\footnote{Department of Mathematics, City University of Hong Kong, Hong Kong, People's Republic of China (kunlun.qi@my.cityu.edu.hk).}}    
\date{}

\maketitle

%\todo{Irene: need address and support.}

\begin{abstract}
This paper is to study the inelastic Boltzmann equation without Grad's angular cutoff assumption, where the well-posedness theory of the solution to the initial value problem is established for the Maxwellian molecules in a space of probability measure defined by Cannone-Karch in [\textit{Comm.~Pure.~Appl.~Math.} \textbf{63} (2010), 747-778] via Fourier transform and the infinite energy solutions are not a priori excluded as well. Meanwhile, the geometric relation of the inelastic collision mechanism is introduced to handle the strong singularity of the non-cutoff collision kernel. Moreover, we extend the self-similar solution to the Boltzmann equation with infinite energy shown by Bobylev-Cercignani in [\textit{J.~Stat.~Phy.} \textbf{106} (2002), 1039-1071] to the inelastic case by a constructive approach, which is also proved to be the large-time asymptotic steady solution with the help of asymptotic stability result in a certain sense.
\end{abstract}

{\small 


}

\tableofcontents

\section{Introduction}
\label{sec:intro}

\textit{1.1 The inelastic Boltzmann equation.}
In recent years, the kinetic equations have been widely used in the granular materials and other industrial applications, where the interactions are described by inelastic collisions\cite{Cercignani1995,Villani2006granularmaterials}. Hence, in this paper, we consider the inelastic homogeneous Boltzmann equation in $ \mathbb{R}^{3} $,
\begin{equation}\label{fQ}
\partial_{t} f(t,v) = Q_{e}(f,f) (t,v),
\end{equation}
with the non-negative initial condition,
\begin{equation}\label{fQinitial}
f(0,v) = F_{0}(v),
\end{equation}
where the unknown $ f=f(t,v) $ is regarded as the density function of a probability distribution, or more generally, a probability measure; and the initial datum $ F_{0} $ is also assumed to be a non-negative probability measure on $ \mathbb{R}^{3} $. The right hand side of \eqref{fQ} is the inelastic Boltzmann collision operator, which is more conveniently defined in the weak formulation \cite{GambaDiffusively} that 
\begin{equation}\label{weak}
\begin{split}
&\int_{\bR^{3}} Q_{e}(g,f)(v)\phi(v) \rd v\\
=& \int_{\bR^{3}} \int_{\bR^{3}} \int_{\Sd^{2}} B(|v-v_*|,\sigma) g(\vs) f(v) \left[ \phi(v')  - \phi(v) \right] \,\rd \sigma \,\rd \vs \,\rd v\\
=& \frac{1}{2}\int_{\bR^{3}} \int_{\bR^{3}} \int_{\Sd^{2}} B(|v-v_*|,\sigma) g(\vs) f(v) \left[ \phi(v') + \phi(\vs') - \phi(v) - \phi(\vs) \right] \,\rd \sigma \,\rd \vs \,\rd v
\end{split}
\end{equation}
where $ \phi(v) $ is a test function, and $ e\in \left[0,1\right] $ is the so-called restitution coefficient ($ e=1 $ denotes elastic collision and $ e=0 $ denotes sticky collision), which is common to be chosen as constant \cite{Villani2006granularmaterials};
the main advantage of the particular weak form is that the inelastic collision law can be only manifested in the test function $ \phi(v') $ and $ \phi(\vs') $, where the post-collisional velocities $ v', v'_{*} $ (with $ v,\vs $ taken as the pre-collisional velocities) including restitution coefficient $ e $ are 
\begin{equation}\label{ve}
\left\{
\begin{array}{lr}
v' =  \frac{v+\vs}{2} + \frac{1-e}{4}(v-\vs) + \frac{1+e}{4}|v-\vs|\sigma &\\
v'_{*} = \frac{v+\vs}{2} - \frac{1-e}{4}(v-\vs) - \frac{1+e}{4}|v-\vs|\sigma.&
\end{array}
\right.
\end{equation}

%\footnote{For the sake of completeness, the detailed derivation of the weak form is included in the appendix \ref{sub:inelasticweak}.} 

%\begin{equation}
%Q_{e}(g,f) = \int_{\bR^{3}} \int_{\Sd^{2}} B(|v-v_*|,\sigma) \left[ Jg('\vs)f('v)- g(\vs)f(v) \right] \,\rd \sigma \,\rd \vs,
%\end{equation}
%where the post-collisional velocities $ \vs,v $ (with $ 'v, '\vs $ taken as the pre-collisional velocities) are 
%\begin{equation}
%\left\{
%\begin{array}{lr}
%'v = \frac{v+\vs}{2} - \frac{1-e}{4e}(v-\vs) + \frac{1+e}{4e}|v-\vs|\sigma  &  \\
%'\vs = \frac{v+\vs}{2} + \frac{1-e}{4e}(v-\vs) - \frac{1+e}{4e}|v-\vs|\sigma, &  
%\end{array}
%\right.
%\end{equation}
% since the following discussion will be based on the weak formulation of the collision operator, here we generally denote $ J $ as the Jacobian of the transform from $ (v,\vs) $ to $ ('v ,'\vs) $, and one can refer to \cite{WeiZhang2012} for more detailed representation forms in non-constant restitution coefficient case.
%\begin{equation}
%J = \left|\frac{\partial(\tilde{v} ,\tilde{v}_{*})}{\partial(v,\vs)} \right|
%\end{equation}

\textit{1.2 The collision kernel.}
The collision kernel $B$ is a non-negative function that depends only on $|v-v_*|$ and cosine of the deviation angle $\theta$, whose specific form can be determined from the intermolecular potential using classical scattering theory \cite{Cercignani}. For example, in the case of inverse power law potentials $U(r) = r^{-(\mathrm{s}-1)}, 2< \mathrm{s} < \infty$, where $r$ is the distance between two interacting particles, $B$ can be separated as the kinetic part and angular part:
\begin{equation}\label{Bb}
B(|v-v_*|,\sigma) = b(\cos\theta) \Phi(|v-v_*|), \quad \cos\theta=\frac{\sigma\cdot (v-v_*)}{|v-v_*|},
\end{equation} 
where kinetic collision part $\Phi(|v-v_*|)=|v-v_*|^{\gamma}$, $\gamma = \frac{\mathrm{s}-5}{\mathrm{s}-1}$, includes hard potential $ (\gamma>0) $, Maxwellian molecule $ (\gamma =0) $ and soft potential $ (\gamma<0) $. Besides, the angular collision part $b(\cos\theta)$ is an implicitly defined function, asymptotically behaving as, when $\theta\rightarrow 0^{+}$,
\begin{equation}\label{noncutoffnu}
\sin \theta b(\cos\theta) \big|_{\theta\rightarrow 0^{+}} \sim K\theta^{-1-\nu}, \quad \nu = \frac{2}{s-1}, \quad 0<\nu<2 \quad \text{and} \quad  K >0,
\end{equation}
i.e., it has a \textit{non-integrable singularity} when the deviation angle $ \theta $ is small. The kernel (\ref{Bb}) encompasses a wide range of potentials, among which we mention two extreme cases: $\mathrm{s}=\infty$, $\gamma=1$, $\nu=0$ corresponds to the hard spheres, and $\mathrm{s}=2$, $\gamma=-3$, $\nu=2$ corresponds to the Coulomb interaction \cite{Villani02}.

Here we will consider the \textit{Maxwellian kernel} $ B(|v-v_{*}|, \sigma) = b\left(\frac{v-v_{*}}{|v-v_{*}|}\cdot \sigma\right)  = b(\cos\theta) $, which implies that $ B $ does not depend on $ |v-v_{*}| $.
%, such that the inelastic Boltzmann collision operator can be rewritten as
%\begin{equation}
%Q_{e}(g,f) = \int_{\bR^{3}} \int_{\Sd^{2}} b(\cos\theta) \left[ Jg(\tilde{v}_{*})f(\tilde{v})- g(\vs)f(v) \right] \rd \sigma \rd \vs
%\end{equation}
The range of deviation angle $ \theta $, namely the angle between pre- and post-collisional velocities, is a full interval $ \left[0,\pi\right] $, but it is customary to restrict it to $ \left[0,\pi/2\right] $ mathematically, replacing $ b(\cos\theta) $ by its ``symmetrized" version \cite{morimoto2016measure}:
\begin{equation}
\left[ b(\cos\theta) + b(\cos\left(\pi-\theta\right))\right]\mathbf{1}_{0\leq \theta \leq \frac{\pi}{2}}.
\end{equation}
As it has been long known, the main difficulty in establishing the well-posedness result for Boltzmann equation is that the singularity of the collision kernel $ b $ is not locally integrable in $ \sigma\in\Sd^{2} $. To avoid this, Harold Grad gave the integrable assumption \cite{GradCutoff} on the collision kernel by a ``cutoff " near singularity. However, here we introduce the full singularity condition for the collision kernel with \textit{non-cutoff assumption},
\begin{equation}\label{noncutoffb}
\exists \alpha_{0} \in (0,2], \quad \text{such that} \quad \int_{0}^{\frac{\pi}{2}} \sin^{\alpha_{0}}\left(\frac{\theta}{2}\right) b(\cos\theta) \sin\theta \rd \theta< \infty,
\end{equation}
which can handle the strongly singular kernel $ b $ in \eqref{noncutoffnu} with some $ 0 < \nu < 2 $ and $ \alpha_{0}\in(\nu, 2 ] $. Besides, we further illustrate that the non-cutoff assumption \eqref{noncutoffb} can be rewritten as
\begin{equation}\label{noncutoffs}
(1-s)^{\frac{\alpha_{0}}{2}} b(s) \in L^{1}[0,1), \quad \text{for $\alpha_{0}\in\left(0,2\right]$},
\end{equation}
by means of the transformation of variable $ s = \cos\theta $ in the symmetric version of $ b $. As mentioned in \cite[Remark 1]{morimoto2012remark}, the full \textit{non-cutoff assumption} \eqref{noncutoffb}, or equivalently \eqref{noncutoffs}, is the extension of the mild non-cutoff assumption of the collision kernel $ b $ used in \cite{cannone2010infinite}, namely,
\begin{equation}\label{mildcutoff}
\left( 1 - s \right)^{\frac{\alpha_{0}}{4}}\left( 1 + s \right)^{\frac{\alpha_{0}}{4}} b(s) \in L^{1}\left(-1,1\right), \quad \text{for $\alpha_{0}\in\left(0,2\right]$}. 
\end{equation}

\textit{1.3 Conservative and dissipative law.}
We also introduce another type of representation for the post-collisional velocities $ v'$ and $\vs' $, that is called the $ \omega $-form, 
\begin{equation}
\left\{
\begin{array}{lr}
v' = v - \frac{1+e}{2} \left[ (v-\vs)\cdot \omega \right]\omega &  \\
\vs' = \vs + \frac{1+e}{2} \left[ (v-\vs)\cdot \omega \right]\omega, &  
\end{array}
\right.
\end{equation}
from which, we can easily verify the conservation of momentum and dissipation of energy:
\begin{equation}
v+\vs = v'+\vs', \quad |v'|^{2} + |\vs'|^{2} - |v|^{2} - |\vs|^{2} = -\frac{1-e^{2}}{2} \left[ (v-\vs)\cdot \omega \right] \leq 0.
\end{equation}
%Note that in this case $ '\vs, 'v $ do not coincide with $ \vs', v' $ since the inelastic collisions are not revertible. Also note that,
Moreover, we also have
\begin{equation}\label{QEcon}
\int_{\bR^{3}} Q_{e}(f,f)(v) \,\rd v = 0, \quad \int_{\bR^{3}} Q_{e}(f,f)(v) v \,\rd v = 0,
\end{equation}
but
\begin{equation}\label{QEdiss}
\int_{\bR^{3}} Q_{e}(f,f)(v) |v|^{2} \,\rd v \leq 0.
\end{equation}

\section{Main Results}
\label{sec:mainresults}

\subsection{Motivation}
Although in the last decades the granular materials has become a popular subject in physical research (for more detailed physical introduction to the kinetic equation in granular material, we refer to \cite{Brilliantov2004}), the mathematical kinetic theory of granular gases is still young and restrictive. For the inelastic Boltzmann equation, most of results are shown in the frame work of \textit{Grad's cutoff assumption} (mainly collision kernel $ b $ is constant) to best knowledge of the author. The three dimensional inelastic Boltzmann equation with \textit{Maxwellian kernel} was first studied by Bobylev-Carrillo-Gamba in \cite{Bobylev2000inelastic}, where the well-posedness theory has been established.  On the other hand, there are lots of work for the so-called \textit{inelastic hard sphere model} as well, where the collision kernel is modified by multiplying the Maxwellian kernel $ b $ with the function of relative velocity. For this model, we refer to the a series of complete work \cite{MM2006hardsphere1,MM2006hardsphere2} by Mischler-Mouhot, where they systematically studied the existence, uniqueness and tail behavior for \textit{inelastic hard sphere} but still with constant angular part $ b $. Besides, some relevant non-constant restitution model \cite{AL2010,AL2014CMP} or Vlasov-Poisson-Boltzmann System \cite{CH2014} are referred for more detailed physical motivation for the inelastic model.

Hence, our first contribution here is expected to systematically establish the well-posed theory of the complete inelastic Boltzmann equation with long-range interaction, handling the \textit{non-cutoff assumption} \eqref{noncutoffb}, if the initial datum is a probability measure (since $ f $ in \eqref{fQ} itself is density function, it is natural to consider the measure valued solution). As usual, we first recall some classical work in the elastic case: starting from late 1990s, Toscani and coauthors have systematically studied the elastic homogeneous equation with finite energy in \cite{Toscani1999, Toscani1995,ToscaniVillani1999}.
In \cite{cannone2010infinite}, Cannone-Karch presented the existence and uniqueness of elastic Boltzmann equation with \textit{Maxwellian molecule} in a space of probability measure defined via Fourier transform, which didn't exclude infinite energy solution, but merely handled the mild singularity of collision kernel. Fortunately, Morimoto extended their results to the strong singularity as well as proving some smoothing effect in \cite{morimoto2012remark}. Meanwhile, Lu-Mouhot showed existence of weak measure valued solution without angular cutoff for hard potential, having finite mass and energy, as well as strong stability and uniqueness under cutoff assumption in \cite{LuMouhot2012,LuMouhot2015}. In more general non-cutoff case (including hard potential and soft potential, finite energy and infinite energy),  Cho-Morimoto-Wang-Yang also studied the measure valued solution with corresponded moment and smoothing  property in their series paper \cite{morimoto2016measure, MWY2015,CMWY2016}. 

Another attractive aspect of the Boltzmann equation is its self-similarity properties, especially in the sense of asymptotic state. Precisely speaking: (i) In the regime of elastic case with Maxwellian kernel, the well-known H-theorem implies the solution to Boltzmann equation tends to the Maxwellian equilibrium as time goes to infinity, if the initial energy is finite. However, when initial energy is infinite, the asymptotic state shall be described by the self-similar solution firstly obtained by Bobylev-Cercignani in \cite{BC2002selfsimilarapplication} and the asymptotic convergence has been proved by Cannone-Karch \cite{cannone2010infinite} and Morimoto-Yang-Zhao \cite{MYZ2017convergence} in the weak and strong sense respectively.
(ii) For the inelastic Boltzmann equation, Bobylev-Cercignani and Bisi-Carrillo-Toscani studied self-similar solutions, long-time behavior respectively in \cite{BC2003} and \cite{BCT2006} for the \textit{cutoff Maxwellian kernel}.
Besides, the convergence to self-similarity for the \textit{inelastic cutoff hard sphere} was further proved by Mischler-Mouhot in
\cite{MM2009inelasticlimit}. More recently, 
Bobylev-Cercignani-Gamba in \cite{BCG2008lecture,BCG2009selfsimilar} developed a more general approach to prove a family of self-similar solutions in radially symmetric case and in \cite{BLM2015electronicProb} Federico-Lucia-Daniel analyzed the long-time asymptotic behavior for \textit{inelastic cutoff Maxwellian kernel} by the probabilistic method, where we also refer to good summary about the convergence results under various circumstances in \cite[Sec 1.2]{BLM2015electronicProb}. Based on the existed work, our contribution in this part is to develop a constructive method in proving the existence of self-similar solution to the inelastic Boltzmann equation with certain singular collision kernel, which attracts all solutions with specific initial conditions in the sense of our defined norm. 

Apart from the work mentioned above, we refer to classical review by Villani \cite{Villani02} for further references in cutoff case and the recent review by Alexandre  \cite{Alexreview2009} under non-cutoff assumption.

\subsection{Main Theorems}

Considering that any solution to be found is a probability measure for any $ t \geq 0 $ after normalization, we 
denote $ P_{0}(\mathbb{R}^{3}) $ as the set of all positive probability measures on $ \mathbb{R}^{3} $ and further $ P_{\alpha}(\mathbb{R}^{3}) $ as the set of probability measures on $ \mathbb{R}^{3} $ with finite moments up to the order $ \alpha\in[0,2] $, which implies the possible existence of infinite energy solution, more precisely,
\begin{equation}\label{Palpha}
\begin{split}
P_{\alpha}(\mathbb{R}^{3}) = \{ f\in P_{0}(\mathbb{R}^{3})& \big|  \int_{\mathbb{R}^{3}} f \,\rd v=1, \ \int_{\mathbb{R}^{3}} |v|^{\alpha} f  \,\rd v < \infty\\
&\text{and if}\ \alpha > 1, \int_{\mathbb{R}^{3}} v_{j} f \,\rd v = 0, \ j = 1,2,3 \}
\end{split}
\end{equation}
see more complete definition of measure valued solution in \cite{morimoto2016measure}. Then the space $ \mathcal{K} $ is constructed to include characteristic functions, see Definition \ref{definitionK}, which consists of the Fourier transformation of probability measures thanks to the Bochner Theorem \cite{cannone2013selfsimilar}. 

Let the Fourier transform of $ f $ be defined by
\begin{equation}
\p(t,\xi) := \mathcal{F}(f)(t,\xi) = \int_{\mathbb{R}^{3}} \e^{-\im v\cdot\xi} f(t,v) \,\rd v,
\end{equation}
it follows that the ``inelastic" version Bobylev identity\footnote{For the sake of completeness, the rigour proof of this identity is presented in the appendix \ref{sub:fourier+}.} can be written as,
\begin{equation}\label{IBE}
\partial_{t} \varphi(t,\xi) = \int_{\Sd^{2}} b\left(\frac{\xi\cdot\sigma}{|\xi|}\right) \left[ \varphi(t,\xie^{+},)\varphi(t,\xie^{-}) - \varphi(t,\xi)\varphi(t,0) \right] \,\rd\sigma,
\end{equation}
where, unlike the elastic case, the $ \xi^{+} $ and $ \xi^{-} $ are defined as
\begin{equation}
\xie^{+} = \frac{\xi}{2} + \frac{1-e}{4}\xi + \frac{1+e}{4}|\xi|\sigma, \quad
\xie^{-} = \frac{\xi}{2} - \frac{1-e}{4}\xi - \frac{1+e}{4}|\xi|\sigma.
\end{equation}
For the sake of convenience, we introduce shorthand parameters $ a_{+} = \frac{1+e}{2} $ and $ a_{-} = \frac{1-e}{2} $, such that, 
\begin{equation}\label{xie+}
\xie^{+} = \left( \frac{1}{2} + \frac{\am}{2} \right)\xi + \frac{\ap}{2}|\xi|\sigma,
\end{equation}
\begin{equation}\label{xie-}
\xie^{-} = \left( \frac{1}{2} - \frac{\am}{2} \right)\xi - \frac{\ap}{2}|\xi|\sigma,
\end{equation}
%\begin{equation}
%\left\{
%\begin{aligned}
%\xie^{+} = \left( \frac{1}{2} + \frac{\am}{2} \right)\xi + \frac{\ap}{2}|\xi|\sigma,\label{xie+} &  \\
%\xie^{-} = \left( \frac{1}{2} - \frac{\am}{2} \right)\xi - \frac{\ap}{2}|\xi|\sigma, \label{xie-}&  
%\end{aligned}
%\right.
%\end{equation}

\begin{equation}
\xie^{+} + \xie^{-} = \xi, \quad |\xie^{+}|^{2} + |\xie^{-}|^{2} = \frac{1+\ap^{2}+ \am^{2}}{2}|\xi|^{2} + \ap\am|\xi|^{2} \frac{\xi\cdot\sigma}{|\xi|}.
\end{equation}

\begin{remark}
	To check this, one can compare \eqref{xie+}-\eqref{xie-} with the elastic case by selecting $ e=1 $, which implies that $ \ap=1, \am=0 $, then 
	\begin{equation}
	\xi^{+} = \frac{\xi + |\xi|\sigma}{2}, \quad \xi^{-} = \frac{\xi - |\xi|\sigma}{2}, \quad |\xi^{+}|^{2} + |\xi^{-}|^{2} =|\xi|^{2}. \notag
	\end{equation}
	which is consistent with the well-known relations of elastic collision.
\end{remark}

Therefore, benefiting from the simple form of Bobylev identity, here our main object will be the equation \eqref{IBE} associated with the following initial condition:
\begin{equation}\label{initial}
\p(0,\xi) = \p_{0}(\xi) = \int_{\mathbb{R}^{3}} \e^{-\im v\cdot \xi} \rd F_{0}(v),
\end{equation}
where if $ \p_{0} \in \mathcal{K}^{\alpha} $ defined as \eqref{Kalpha} is the Fourier transform of a probability measure $ F_{0} $ satisfying \eqref{Palpha}, then the corresponding solution $ \p= \p(t,\xi) $ to \eqref{IBE}-\eqref{initial} is the Fourier transform of a solution $ f = f(t,v) $ to the original initial value problem \eqref{fQ}-\eqref{fQinitial}, see more explanations in \cite{cannone2010infinite}.

Now we are in a position to state our main theorem on the well-posedness of the solution $ \p $ to the initial value problem \eqref{IBE}-\eqref{initial}.

\begin{theorem} \label{noncutoffwellposedness}
	\emph{(Well-posedness under non-cutoff assumption)} Assume that $ e\in (0,1] $ and the collision kernel $ b $ satisfies the non-cutoff assumption \eqref{noncutoffb} for some $ \alpha_{0}\in\left[0,2\right] $, then for each $ \alpha\in\left[\alpha_{0},2\right]$ and initial condition $ \p_{0}\in\K^{\alpha} $, there exists a solution $ \p\in C\left(  \left[0,\infty\right), \K^{\alpha} \right)$ to the initial value problem \eqref{IBE}-\eqref{initial} and the solution $ \p $ is unique in the space $ C\left(  \left[0,\infty\right), \K^{\alpha_{0}} \right) $. \\
	Furthermore, for two solutions $ \p, \tilde{\p} \in C\left( \left[0,\infty\right), \K^{\alpha} \right) $ corresponding to the initial datum $ \p_{0}, \tilde{\p}_{0} $ respectively, we have the stability result, for every $ t \geq 0 $,
	\begin{equation}\label{stabilitynoncutoff}
	\left\| \p(t,\cdot)-\tilde{\p}(t,\cdot) \right\|_{\alpha} \leq \e^{\lambda_{e,\alpha}t} \left\| \p_{0}-\tilde{\p}_{0} \right\|_{\alpha},
	\end{equation}
	where the finite parameter $ \lambda_{e,\alpha} $ is defined as,
	\begin{equation}\label{lamdadefinitionmain}
	\lambda_{e,\alpha} \equiv \int_{\Sd^{2}} b\left(\frac{\xi\cdot\sigma}{|\xi|}\right) \left(\frac{|\xie^{+}|^{\alpha}+|\xie^{-}|^{\alpha}}{|\xi|^{\alpha}} -1 \right) \,\rd \sigma.
	\end{equation}
\end{theorem}
Note that the quantity will $ \lambda_{e,\alpha} $ appears systematically in the rest of the paper, which nearly play the same role as corresponded parameter $ \lambda_{\alpha} $ in elastic case \cite{cannone2010infinite}, defined by
\begin{equation}
\lambda_{\alpha} \equiv \int_{\Sd^{2}} b\left(\frac{\xi\cdot\sigma}{|\xi|}\right) \left(\frac{|\xi^{+}|^{\alpha}+|\xi^{-}|^{\alpha}}{|\xi|^{\alpha}} -1 \right)\rd \sigma = 2\pi \int_{0}^{\frac{\pi}{2}} b(\cos\theta) \left(\sin^{\alpha}\frac{\theta}{2} + \cos^{\alpha}\frac{\theta}{2} - 1 \right) \sin\theta\rd\theta.
\end{equation}
and $ \lambda_{e,\alpha} = \lambda_{\alpha}$ if and only if the restitution coefficient $ e=1 $. More important properties of $ \lambda_{e,\alpha} $ and another parameter $ \gamma_{e,\alpha} $ as \eqref{gammae} will be discussed in the Lemma \ref{parametergl} below.

The complete proof of Theorem \ref{noncutoffwellposedness} will be presented in section \ref{sec:noncutoff} by a delicate compact argument, which is based on the well-posed theory under cutoff assumption firstly given in the section \ref{sec:cutoff}. The uniqueness conclusion is guaranteed by stability result under non-cutoff assumption \eqref{noncutoffb}.

Besides that, in order to study the large time behaviour of a class of solution to system \eqref{IBE}-\eqref{initial},
we also consider the self-similar scaling $ \p\left(\xi,t\right) =  \Phi\left( \xi \e^{\mu t} \right)$ such that we can reduce the study of self-similar solution to the study of stationary solution to the following rescaled equation:
\begin{equation}\label{Phi1}
\mu \eta \cdot \nabla \Phi(\eta) = \int_{\Sd^{2}} b\left(\frac{\eta\cdot\sigma}{|\eta|}\right) \left[ \Phi(\eta_{e}^{+})\Phi(\eta_{e}^{-}) - \Phi(\eta)\Phi(0)\right]\rd\sigma,
\end{equation}
which is obtained by substituting the profile $ \Phi\left( \xi \e^{\mu t} \right) $ into \eqref{IBE} and the variable $ \eta^{+} $ and $ \eta^{-} $ have the analogous definition as the vector in \eqref{xie+}-\eqref{xie-}. In fact, we claim that the coefficients $ \mu $ can be determined by ;
\begin{equation}\label{muea}
\mu = \mu_{e,\alpha} = \frac{\lambda_{e,\alpha}}{\alpha},
\end{equation}
which will be shown in the proof of the following Theorem \ref{SteadyExistence} in section \ref{sec:selfsimilar}. In contrast with the general method in \cite{BCG2008lecture,BCG2009selfsimilar}, we apply a totally different constructive approach to obtain the self-similar solution for the inelastic Boltzmann equation with infinite energy motivated by \cite{BC2002selfsimilarapplication}.

\begin{theorem}\label{SteadyExistence}
	\emph{(Existence of self-similar solution)} Assume that $ e\in (0,1] $ and the collision kernel $ b $ satisfies the non-cutoff assumption \eqref{noncutoffb} for some $ \alpha\in\left(0,2\right) $. For each constant $ K <0 $ and $ \mu_{e,\alpha}$ defined in \eqref{muea}, there exists a radially symmetric solution $ \Phi(\eta) = \Phi(\left| \eta \right|) = \Phi^{(\alpha)}_{e,K} \in \mathcal{K}^{\alpha} $ to the equation \eqref{Phi1} satisfying 
	\begin{equation}\label{Phiasym}
	\lim\limits_{\left| \eta \right|\rightarrow 0 }\frac{ \Phi^{(\alpha)}_{e,K}\left(\eta\right)-1}{\left| \eta \right|^{\alpha}} = K,
	\end{equation}
	where $ K $ is the coefficient $ \Psi_{1}^{(\alpha)} $ of \eqref{Psiconclusion}.
\end{theorem}
The complete proof will be given in section \ref{sec:selfsimilar}.
\begin{remark}
	Note that the negativity of constant $ K $ is definite, though its value is not strictly determined, which has the same reason as elastic case that $ \Phi^{(\alpha)}_{e,K} $ is proved to be characteristic function satisfying $ \left|\Phi^{(\alpha)}_{e,K}(\eta) \leq 1 \right| $ as well, see \cite[Remark 6.3]{cannone2010infinite} for more details.
\end{remark}

On the other hand, it is more convenient to work in self-similar variables to study the role that the self-similar profile $ \Phi $ plays in large time behavior, which means that, given a solution $ \p\left(t,\xi\right) $, we consider another new function
\begin{equation}
\phi_{e}^{(\alpha)}\left(t,\xi\right) = \p\left( t, \xi\e^{-\mu_{e,\alpha} t}\right), 
\end{equation}
therefore, we can reduce the initial value problem \eqref{IBE}-\eqref{initial} to the following new initial value problem,
\begin{equation}\label{phi}
\partial_{t} \phi_{e}^{(\alpha)}\left(t, \xi\right) + \mu_{e,\alpha} \xi\cdot\nabla\phi_{e}^{(\alpha)}\left(t, \xi\right) = \int_{\Sd^{2}} b\left(\frac{\xi\cdot\sigma}{|\xi|}\right) \left[ \phi_{e}^{(\alpha)}\left(t,\xie^{+}\right)\phi_{e}^{(\alpha)}\left(t,\xie^{-}\right) - \phi_{e}^{(\alpha)}\left(t,\xi\right)\phi_{e}^{(\alpha)}\left(t,0\right) \right] \rd\sigma,
\end{equation}
with the following initial datum
\begin{equation}\label{phii}
\phi\left(0,\xi\right) = \phi_{0}\left(\xi\right) = \p_{0}\left(\xi\right).
\end{equation}
Note that, in the new variable, the self-similar profiles $ \Phi $ is claimed as stationary solutions to the initial value problem above \eqref{phi}-\eqref{phii}.
Before showing that, we give the following stability result with respect to the rescaled initial value problem \eqref{phi}-\eqref{phii}. 

\begin{theorem}\label{AsymptoticStability}
	\emph{(Asymptotic stability of rescaled equation)} Assume that $ e\in (0,1] $ and the collision kernel $ b $ satisfies the non-cutoff assumption \eqref{noncutoffb} for $ \alpha_{0}\in \left(0,2\right) $. Let $ \alpha\in\left[\alpha_{0},2\right) $ and suppose that the two initial datums $ \phi_{0}, \tilde{\phi}_{0} \in \mathcal{K}^{\alpha} $ satisfy the following condition
	\begin{equation}\label{phiinitial}
	\lim\limits_{|\xi|\rightarrow 0} \frac{\phi_{0}(\xi)-\tilde{\phi}_{0}(\xi)}{|\xi|^{\alpha}} = 0,
	\end{equation}
	then the corresponded solutions $ \phi_{e}^{(\alpha)}(t,\xi), \tilde{\phi}_{e}^{(\alpha)}(t,\xi) $ to the rescaled initial value problem \eqref{phi}-\eqref{phii} approach each other in the following sense:
	\begin{equation}
	\lim\limits_{t\rightarrow \infty} \left\| \phi_{e}^{(\alpha)}(t,\cdot) - \tilde{\phi}_{e}^{(\alpha)}(t,\cdot) \right\|_{\alpha} = 0.
	\end{equation}
\end{theorem}
The complete proof of Theorem \ref{AsymptoticStability} is presented in section \ref{subsec:asymptotic}, which relies on the basic stability result \eqref{stabilitynoncutoff} above.
\begin{remark}
	 It is noted that asymptotic stability result can be reduced to the pre-scaled initial value problem \eqref{IBE}-\eqref{initial}, in the sense that if two initial datum $ \p_{0}, \tilde{\p}_{0} \in \mathcal{K}^{\alpha} $ satisfying the \eqref{phiinitial},
	\begin{equation}
	\lim\limits_{t\rightarrow \infty} \e^{-\lambda_{e,\alpha}t} \left\| \p(t,\cdot)-\tilde{\p}(t,\cdot) \right\|_{\alpha} = 0,
	\end{equation}
	which is the direct consequence after changing variable back, similar to the elastic case, see more in \cite[Remark 2.9-2.10]{cannone2010infinite}.
\end{remark}

Together with the Proposition \ref{theoremx} about $ \Phi $ and the asymptotic stability result Theorem \ref{AsymptoticStability}, we can directly prove that the solution $ \phi_{e}^{(\alpha)}\left(t,\xi \right) = \p\left( t, \xi\e^{-\mu_{e,\alpha} t} \right) $ to \eqref{IBE}-\eqref{initial} converges (in self-similar variables) towards the self-similar profile $ \Phi $ for some specific initial conditions,

\begin{corollary}
Assume that $ e\in (0,1] $ and the collision kernel $ b $ satisfies the non-cutoff assumption \eqref{noncutoffb} for $ \alpha_{0}\in \left(0,2\right) $. Let $ \alpha\in\left[\alpha_{0},2\right) $ and $ \phi_{0}(\xi) $ be the initial condition such that
\begin{equation}
	\lim\limits_{|\xi|\rightarrow 0} \frac{\phi_{0}(\xi)-1}{|\xi|^{\alpha}} = K,
\end{equation}
for some $ K \leq 0 $. Then, the solution $ \phi_{e}^{(\alpha)}(t,\xi) $ to the initial value problem \eqref{phi}-\eqref{phii} converges to the self-similar profile $ \Phi^{(\alpha)}_{e,K} $ in the following sense,
\begin{equation}
\lim\limits_{t\rightarrow \infty} \left\| \phi_{e}^{(\alpha)}(t,\cdot) - \Phi^{(\alpha)}_{e,K} \right\|_{\alpha} = 0, \quad \text{if} \quad K<0,
\end{equation}
and 
\begin{equation}
\lim\limits_{t\rightarrow \infty} \left\| \phi_{e}^{(\alpha)}(t,\cdot) - 1 \right\|_{\alpha} = 0, \quad \text{if} \quad K=0.
\end{equation}
\end{corollary}
\begin{remark}
	Note that the proof of the convergence can be regarded as the special case of Theorem \ref{AsymptoticStability} above, thus, this convergence in the weak sense holds true in the metric of the space $ \mathcal{K}^{\alpha} $. To return the function $ \phi$ and $ \Phi^{(\alpha)}_{e,K} $ in the Fourier space back to  $ f$ and its corresponded steady profile $ F $ in velocity space, it is more appropriate to consider in the space $ \tilde{P}_{\alpha} = \mathcal{F}^{-1}(\mathcal{K}^{\alpha}) $, which is recently introduced by Morimoto-Wang-Yang \cite{MWY2015}. For this part, the convergence result in a more accurate sense is under preparing by the author.
\end{remark}

\subsection{Plan of the paper}
The paper is organized as follows. In the next section \ref{sec:pre} we will first give some basic properties of the characteristic function $ \p $ as well as some useful estimates about inelastic variables $ \xie^{+},\xie^{-} $ and the well-definedness of inelastic collision operator, which are the key parts to further establish well-posedness theory. In section \ref{sec:cutoff}, we construct the solution under cutoff assumption by using the Banach fixed point theorem and further prove the stability result. The well-posed theory under the non-cutoff assumption is established by compactness argument based on cutoff results in section \ref{sec:noncutoff}. The final section \ref{sec:selfsimilar} is devoted to study the  self-similar solution to the inelastic equation for some certain initial conditions and prove the asymptotic convergence to such self-similar profile in a suitable sense.

\section{Preliminary}
\label{sec:pre}

\subsection{Some Properties of Characteristic Function}
\label{subsec:characteristic}

As the original Boltzmann equation \eqref{fQ}-\eqref{fQinitial} has been transformed into the study of the initial value problem in the Fourier variables \eqref{IBE}-\eqref{initial} in the space of characteristic functions $\mathcal{K}$, we first present some basic properties of characteristic function, which has been devoted to the study of spatially homogeneous Boltzmann equation in Fourier space for a long time.

\begin{definition}\label{definitionK}
	A function $ \p := \mathbb{R}^{3} \mapsto \mathbb{C} $ is called a characteristic function if there is a probability measure $ F $ (i.e. a Borel measure with $ \int_{\mathbb{R}^{3}} \rd F(v) = 1 $) such that we have the identity $ \p(\xi)= \hat{f}(\xi) = \int_{\mathbb{R}^{3}} \e^{-\im v\cdot \xi} \rd F(v) $. We will denote the set of all characteristic function $ \p := \mathbb{R}^{3} \mapsto \mathbb{C} $ by $\mathcal{K}$.
\end{definition}

Inspired by \cite{cannone2010infinite,morimoto2016measure}, we also define the subspace $ \K^{\alpha} $ of all characteristic functions $ \K $ as following:
\begin{equation}\label{Kalpha}
\K^{\alpha} = \left\lbrace \p\in\K; \left\| \p-1 \right\|_{\alpha}<\infty \right\rbrace,
\end{equation}
where 
\begin{equation}
\left\| \p-1 \right\|_{\alpha} = \sup_{\xi\in\mathbb{R}^{3}} \frac{\left| \p(\xi)-1 \right|}{|\xi|^{\alpha}}.
\end{equation}
The set $ \mathcal{K}^{\alpha} $ endowed with the distance $ \left\| \cdot \right\|_{\alpha} $, for any $ \p,\tilde{\p}\in\mathcal{K}^{\alpha} $, 
\begin{equation}
\left\| \p-\tilde{\p} \right\|_{\alpha} = \sup_{\xi\in\mathbb{R}^{3}} \frac{\left| \p(\xi)-\tilde{\p}(\xi) \right|}{|\xi|^{\alpha}},
\end{equation}
is a complete metric space, with the following embedding relation,
\begin{equation}
\{1\} \subseteq \mathcal{K}^{\alpha} \subseteq \mathcal{K}^{\alpha_{0}} \subseteq \mathcal{K}^{0}, \quad \text{for all} \quad 2 \geq \alpha \geq \alpha_{0} \geq 0.
\end{equation}
Note that the Fourier transform of every probability measure in $ P_{\alpha}(\mathbb{R}^{3}) $ belongs to $ \mathcal{K}^{\alpha} $, however, the set $ \mathcal{K}^{\alpha} $ is bigger than the $ \mathcal{F}(P_{\alpha}) $, see \cite[Remark 3.16]{cannone2010infinite}.

\begin{lemma}\label{L1}
For any positive definite function $ \p=\p(\xi) \in \mathcal{K} $ such that $ \p(0) = 1 $, we have
\begin{equation}\label{l1}
\left|\p(\xi) - \p(\eta)\right|^{2} \leq 2\left( 1 - \Re\left[ \p(\xi-\eta) \right]\right) 
\end{equation}
and
\begin{equation}\label{l2}
\left|\p(\xi)\p(\eta) - \p(\xi+\eta) \right|^{2} \leq \left(1-\left|\p(\xi) \right|^{2 }\right)\left(1-\left|\p(\eta) \right|^{2 }\right)
\end{equation}
for all $ \xi,\eta\in\mathbb{R}^{3} $ and moreover if $ \p \in \mathcal{K}^{\alpha} $, then 
\begin{equation}\label{l3}
\left|\p(\xi) - \p(\xi+\eta) \right| \leq \left\| \p-1 \right\|_{\alpha}\left( 4 |\xi|^{\frac{\alpha}{2}} |\eta|^{\frac{\alpha}{2}} + |\eta|^{\alpha}\right).
\end{equation}
\end{lemma}
\begin{proof}
	The proof is based on the definition of positive definite function, where the inequalities \eqref{l1}-\eqref{l2} can be found in \cite[Lemma 3.8]{cannone2010infinite} and the last inequality \eqref{l3} can be found in \cite[Lemma 2.1]{morimoto2012remark} for reference.
\end{proof}

\begin{lemma}\label{reim}
	Let $ \alpha\in \left[ 0, 2 \right] $ and $ \p\in\K^{\alpha} $, then $ \Re(\p)\in\K^{\alpha} $,
	\begin{equation}\label{re}
	\left\| \Re(\p) - 1 \right\|_{\alpha} \leq \left\| \p - 1 \right\|_{\alpha}, 
	\end{equation}
	and 
	\begin{equation}\label{im}
	\sup_{\xi\in\mathbb{R}^{3}/\{0\}} \frac{\left|\Im[\p(\xi)] \right|}{|\xi|^{\alpha}} \leq \left\| \p-1\right\|_{\alpha}.
	\end{equation}
\end{lemma}
\begin{proof}
	In fact, for any characteristic function $ \p\in\K^{\alpha} $, its real part $ \Re(\p) $ is the characteristic function as well, thanks to the identity $ \Re(\p) = \left( \p + \bar{\p}\right)/2 $. Then, by the Pythagorean theorem, we have 
	\begin{equation}
	\left|\p(\xi) - 1 \right|^{2} = \left| \Im\left[ \p(\xi)\right] \right|^{2} + \left| \Re\left[ \p(\xi)-1\right]\right|^{2} \geq \left|\Re\left[\p(\xi)\right] -1 \right|^{2}.
	\end{equation}
	After dividing the equation above by $ |\xi|^{\alpha} $ and calculating the supremum with respect to $ \xi\in\mathbb{R}^{3}/\{0\}  $, we obtain 
	\begin{equation}
	\left\| \p - 1 \right\|_{\alpha} \geq \left\| \Re(\p) - 1 \right\|_{\alpha}.
	\end{equation}
	Besides, considering the inequality $ \left|\p(\xi) -1 \right| \geq \left|\Im\p(\xi) \right| $, we can find that 
	\begin{equation}
	\sup_{\xi\in\mathbb{R}^{3}/\{0\}} \frac{\left|\Im[\p(\xi)] \right|}{|\xi|^{\alpha}} \leq \left\| \p-1\right\|_{\alpha}.
	\end{equation}
\end{proof}

\subsection{Useful Estimates about Inelastic Variables and Collision Operator}
\label{subsec:elasticvariable}

In this subsection, we will introduce some technical estimates of variable $ \xie^{+} $ and $ \xie^{-} $ in the following Lemma \ref{bound} and \ref{well}, based on our observation and some elementary inequalities, which then play a key role in proving that the inelastic Bobylev Identity is also well-defined under non-cutoff assumption \eqref{noncutoffb} in Lemma \ref{wellsym}.
\begin{lemma}\label{bound} 
	Let $ \xie^{+} $ and $ \xie^{-} $ be the variables defined as \eqref{xie+} and \eqref{xie-} respectively with $ e\in (0,1] $, then for $ \alpha\in[0,2] $, we have
	\begin{equation}\label{bound+}
	\left[\ap(1+\am)\right]^{\frac{\alpha}{2}} \left( \frac{1+\frac{\xi\cdot\sigma}{|\xi|}}{2}\right)^{\frac{\alpha}{2}} \left|\xi\right|^{\alpha} \leq \left|\xie^{+}\right|^{\alpha} \leq \left[ \frac{\left(1+\am\right)^{2}+ \left(\ap\right)^{2}}{2}\right]^{\frac{\alpha}{2}} \left( \frac{1+\frac{\xi\cdot\sigma}{|\xi|}}{2}\right)^{\frac{\alpha}{2}} \left|\xi\right|^{\alpha},
	\end{equation}
	and
	\begin{equation}\label{bound-}
	\left|\xie^{-}\right|^{\alpha} = \left(\ap^{2}\right)^{\frac{\alpha}{2}} \left( \frac{1-\frac{\xi\cdot\sigma}{|\xi|}}{2}\right)^{\frac{\alpha}{2}} \left|\xi\right|^{\alpha}.
	\end{equation}
\end{lemma}
\begin{proof}
	The proof is based on the observation as well as the Cauchy's inequality: Starting from the specific form $ \xie^{+} $ defined as \eqref{xie+} and calculating the identity $ |\xie^{+}|^{2} = \xie^{+}\cdot\xie^{+} $, we have
	\begin{equation}
	|\xie^{+}|^{2} = \left[ \left(\frac{1+\am}{2}\right)^{2} + \left( \frac{\ap}{2}\right)^{2} + \frac{\ap(1+\am)}{2}\frac{\xi\cdot\sigma}{|\xi|} \right]\left|\xi\right|^{2},
	\end{equation}
	moreover, considering the Cauchy's inequality $ \frac{\ap(1+\am)}{2} \leq \frac{(1+\am)^{2}}{4} + \frac{\ap^{2}}{4} = \left(\frac{1+\am}{2}\right)^{2} + \left( \frac{\ap}{2}\right)^{2} $, we're able to extract the common factor $ \left[ \frac{\left(1+\am\right)^{2}+ \left(\ap\right)^{2}}{2}\right] $ and then obtain the right hand side of \eqref{bound+} by computing $ \left(\cdot\right)^{\frac{\alpha}{2}} $,
	\begin{equation}\label{bound+1}
	\begin{split}
	\left|\xie^{+}\right|^{\alpha} \leq& \left[ \left(\frac{1+\am}{2}\right)^{2}+ \left(\frac{\ap}{2}\right)^{2}\right]^{\frac{\alpha}{2}} \left( 1+\frac{\xi\cdot\sigma}{|\xi|}\right)^{\frac{\alpha}{2}} \left|\xi\right|^{\alpha}\\
	=& \left[ \frac{\left(1+\am\right)^{2}+ \left(\ap\right)^{2}}{2}\right]^{\frac{\alpha}{2}} \left( \frac{1+\frac{\xi\cdot\sigma}{|\xi|}}{2}\right)^{\frac{\alpha}{2}} \left|\xi\right|^{\alpha},
	\end{split}
	\end{equation}
	meanwhile, by the same Cauchy's inequality $ \left(\frac{1+\am}{2}\right)^{2} + \left( \frac{\ap}{2}\right)^{2} \geq  \frac{\ap(1+\am)}{2} $, we can obtain the left hand side of \eqref{bound+} by computing $ \left(\cdot\right)^{\frac{\alpha}{2}} $,
	\begin{equation}\label{bound-1}
	\begin{split}
	\left|\xie^{+}\right|^{\alpha} \geq& \left[ \frac{\ap(1+\am)}{2} \right]^{\frac{\alpha}{2}} \left( 1+\frac{\xi\cdot\sigma}{|\xi|}\right)^{\frac{\alpha}{2}} \left|\xi\right|^{\alpha}\\
	=&\left[\ap(1+\am)\right]^{\frac{\alpha}{2}} \left( \frac{1+\frac{\xi\cdot\sigma}{|\xi|}}{2}\right)^{\frac{\alpha}{2}} \left|\xi\right|^{\alpha}.
	\end{split}
	\end{equation}
	The proof of \eqref{bound+} will be complete by combining the \eqref{bound+1} and \eqref{bound-1}. On the other hand, we can also computing $ \left|\xie^{-}\right|^{2} $ by using the formula \eqref{xie-} ,to obtain that
	\begin{equation}
	|\xie^{-}|^{2} = \left[ \left(\frac{1-\am}{2}\right)^{2} + \left( \frac{\ap}{2}\right)^{2} - \frac{\ap(1-\am)}{2}\frac{\xi\cdot\sigma}{|\xi|} \right]\left|\xi\right|^{2},
	\end{equation}
	by noticing the relation between $ \ap $ and $ \am $ that $ \ap = 1-\am $ as well as further calculation, we can directly get the identity \eqref{bound-}.
\end{proof}

\begin{lemma}\label{well}
	Let $ \alpha\in \left[ 0, 2 \right] $ and $ e\in (0,1] $. For each $ \xi\in\mathbb{R}^{3} $, the inelastic variables $ \xie^{+} $ and $ \xie^{-} $ are defined as \eqref{xie+} and \eqref{xie-} with some fixed $ \sigma\in\Sd^{2} $ respectively. Then, for $ \p\in\K^{\alpha} $,
	\begin{equation}\label{1}
	\left|\p(\xie^{+})\p(\xie^{-}) - \p(\xi)\p(0) \right| \leq 4\left|\xie^{+}\right|^{\frac{\alpha}{2}}\left|\xie^{-}\right|^{\frac{\alpha}{2}}\left\|\p-1\right\|_{\alpha},
	\end{equation}
	more precisely, 
	\begin{equation}\label{2}
	\begin{split}
	\qquad &\left|\p(\xie^{+})\p(\xie^{-}) - \p(\xi)\p(0) \right| \\
	&\leq 4 \left(\ap^{2}\right)^{\frac{\alpha}{4}} \left[ \frac{\left(1+\am\right)^{2}+ \left(\ap\right)^{2}}{2}\right]^{\frac{\alpha}{4}} \left( \frac{1-\frac{\xi\cdot\sigma}{|\xi|}}{2}\right)^{\frac{\alpha}{4}}\left( \frac{1+\frac{\xi\cdot\sigma}{|\xi|}}{2}\right)^{\frac{\alpha}{4}} \left|\xi\right|^{\alpha} \left\| \p  -1\right\|_{\alpha}.
	\end{split}
	\end{equation}
\end{lemma}
\begin{proof}
	Start from the following identity
	\begin{equation}\label{p1}
	1 - \left|\p(\xi)\right|^{2} = \big( 1-\p(\xi)\big)\big( 1+ \overline{\p(\xi)}\big) + 2\Im\left[ \p(\xi)\right]\im,
	\end{equation}
	together with the estimate \eqref{im} in Lemma \ref{reim} and the following inequality,
	\begin{equation}
	\left| 1 + \overline{\p(\xi)} \right|\leq 1 + \left| \p(\xi) \right| \leq 2,
	\end{equation} 
	we can deduce from the inequality \eqref{p1} that
	\begin{equation}\label{p2}
	0 \leq 1 - \left|\p(\xi)\right|^{2} \leq 4 \left| \xi \right|^{\alpha} \left\| \p-1 \right\|_{\alpha}.
	\end{equation} 
	In fact, the \eqref{p2} holds if we substitute $ \xie^{+} $ and $ \xie^{-} $ into it. Recalling that $ \p(0) =1  $ and the relation $ \xie^{+} + \xie^{-} = \xi$, consequently, we're able to apply the inequality \eqref{l2}, 
	\begin{align}
	\left| \p(\xie^{+})\p(\xie^{-}) - \p(\xi) \right| & \leq \sqrt{\left(1 - \left|\p(\xie^{+})\right|^{2}\right)\left(1 - \left|\p(\xie^{-})\right|^{2}\right)}\\
	&\leq 4 \left| \xie^{+} \right|^{\frac{\alpha}{2}}\left| \xie^{-} \right|^{\frac{\alpha}{2}}\left\| \p  -1\right\|_{\alpha}.
	\end{align}
	Furthermore, considering the \eqref{bound+} and \eqref{bound-} in Lemma \ref{bound}, we can finally obtain \eqref{2}.
\end{proof}

With the help of the preliminary estimates \eqref{reim} - \eqref{well} above, we're able to prove the following technical Lemma \ref{wellsym} to show that the nonlinear term in the right hand side of \eqref{IBE} is well-defined for any function $ \p\in\K^{\alpha} $, even the strong singularity condition \eqref{noncutoffb} of the collision kernel $ b $ holds.

\begin{lemma}\label{wellsym}
	Assume that $ e\in (0,1] $ and collision kernel $ b $ satisfies the  non-cutoff assumption \eqref{noncutoffb} for $ \alpha_{0}\in(0,2] $. If $ \p\in\mathcal{K}^{\alpha}$ for $ \alpha\in[\alpha_{0},2] $, then
	\begin{equation}\label{well2}
	\begin{split}
	\left| \int_{\Sd^{2}} b\left(\frac{\xi\cdot\sigma}{|\xi|}\right) \left[\p(\xie^{+})\p(\xie^{-}) - \p(0)\p(\xi) \right] \rd\sigma \right| \leq C_{e} \left[ \int_{0}^{\frac{\pi}{2}} \sin^{\alpha}\left(\frac{\theta}{2}\right)b(\cos\theta) \sin\theta \rd \theta \right] \left\| 1-\p\right\|_{\alpha} \left|\xi\right|^{\alpha} < \infty
	\end{split}
	\end{equation}
	where $ C_{e} $ is a constant depending on the restitution coefficient $ e $.
\end{lemma}

\begin{figure}[ht]
	\centering
	\includegraphics[scale=0.7]{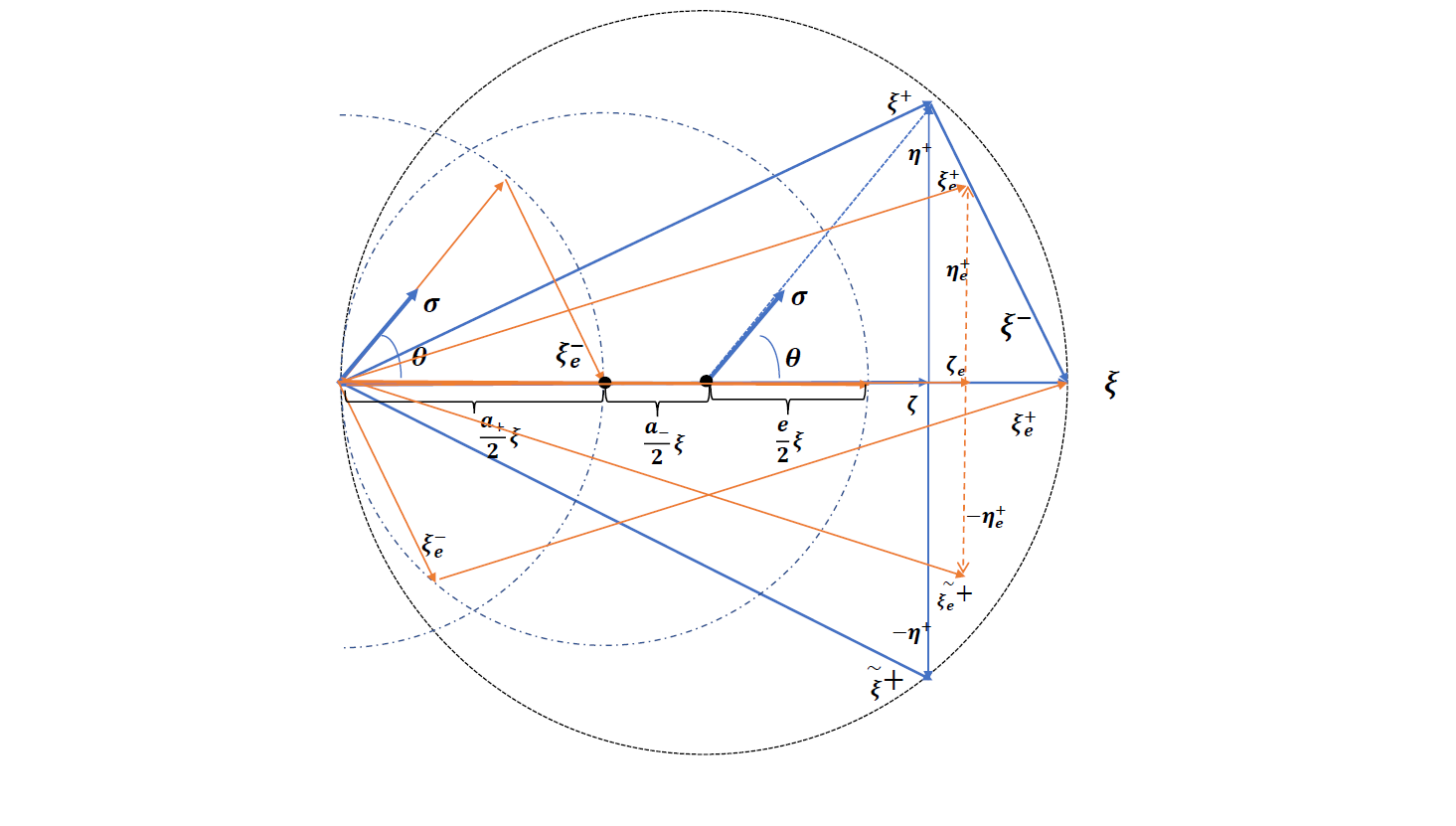}
	\caption{Illustration of the inelastic collision mechanism with $ \cos\theta = \frac{\xi\cdot\sigma}{|\xi|} $ and $ \eta_{e}^{+} = \xie^{+} - \zeta_{e} $.}
	\label{fig:label}
\end{figure}

\begin{proof}
	By introducing $ \zeta_{e} = \left(\xie^{+}\cdot\frac{\xi}{|\xi|}\right) \frac{\xi}{|\xi|} $ as the middle variable as well as considering the fact that $ \p(0) = 1 $,  %and then $ \tilde{\xi}^{+} = \zeta - (\xi^{+} - \zeta) $, which is symmetric to $ \xie^{+} $ on $ \Sd^{2} $, we have 
	\begin{align}
	&\int_{\Sd^{2}} b\left(\frac{\xi\cdot\sigma}{|\xi|}\right) \left[\p(\xie^{+})\p(\xie^{-}) -\p(0) \p(\xi) \right] \rd\sigma \\
	= & \int_{\Sd^{2}} b\left(\frac{\xi\cdot\sigma}{|\xi|}\right) \left[\p(\xie^{+})\p(\xie^{-}) - \p(\xie^{+}) + \p(\xie^{+}) - \p(\xi) \right] \rd\sigma \\
	= & \int_{\Sd^{2}} b\left(\frac{\xi\cdot\sigma}{|\xi|}\right) \left[\p(\xie^{+})  - \p(\xi) \right] \rd\sigma + \int_{\Sd^{2}} b\left(\frac{\xi\cdot\sigma}{|\xi|}\right) \p(\xie^{+})\left[\p(\xie^{-}) - 1 \right] \rd\sigma\\
	= & \frac{1}{2}\int_{\Sd^{2}} b\left(\frac{\xi\cdot\sigma}{|\xi|}\right) \left[\p(\xie^{+}) + \p(\tilde{\xie}^{+}) - 2\p(\xi) \right] \rd\sigma + \int_{\Sd^{2}} b\left(\frac{\xi\cdot\sigma}{|\xi|}\right) \p(\xie^{+})\left[\p(\xie^{-}) - 1 \right] \rd\sigma\\
	= & \frac{1}{2}\int_{\Sd^{2}} b\left(\frac{\xi\cdot\sigma}{|\xi|}\right) \left[\p(\xie^{+}) + \p(\tilde{\xie}^{+}) - 2\p(\zeta_{e}) \right] \rd\sigma + \int_{\Sd^{2}} b\left(\frac{\xi\cdot\sigma}{|\xi|}\right) \left[\p(\zeta_{e}) - \p(\xi) \right] \rd\sigma\\ &+\int_{\Sd^{2}} b\left(\frac{\xi\cdot\sigma}{|\xi|}\right) \p(\xie^{+})\left[\p(\xie^{-}) - 1 \right] \rd\sigma\\
	:= & I_{1} + I_{2} + I_{3}
	\end{align} %$ \xie^{+} =  \xi - \xie^{-} = \zeta + \eta^{+} + \xi^{-} -\xie^{-} $
	
	(i) For the first part $ I_{1} $, by considering the symmetric geometry relation $ \xie^{+} = \zeta_{e} + \eta_{e}^{+} $ and $ \xie^{-} = \zeta_{e} +(-\eta_{e}^{+}) $ as in Figure \ref{fig:label}, we obtain,
	\begin{align}
	&\left| \p(\xie^{+}) +\p(\tilde{\xie}^{+}) - 2\p(\zeta_{e}) \right| = \left| \int_{\mathbb{R}^{3}} \e^{-\im \zeta_{e}\cdot v} \left( \e^{-\im\eta_{e}^{+}\cdot v} + \e^{\im\eta_{e}^{+}\cdot v} -2  \right) \rd F(v) \right|\\
	\leq &  \int_{\mathbb{R}^{3}} \left|\e^{-\im \zeta_{e}\cdot v} \right|  \left( 2 - \e^{-\im\eta_{e}^{+}\cdot v} - \e^{\im\eta_{e}^{+}\cdot v}  \right) \rd F(v) \\
	= & 2 - \p(\eta_{e}^{+}) - \p(-\eta_{e}^{+})\\
	= & \left[ 1- \p(\eta_{e}^{+}) \right] + \left[ 1- \p(-\eta_{e}^{+}) \right]\\
	\leq & 2 \left\| 1-\p\right\|_{\alpha} \left| \eta_{e}^{+}\right|^{\alpha} \leq 2 \left\| 1-\p\right\|_{\alpha} \left| \eta^{+}\right|^{\alpha} \leq 2 \left\| 1-\p\right\|_{\alpha}  |\xi|^{\alpha} \sin^{\alpha}\left(\frac{\theta}{2}\right),
	\end{align}
	%	&\left| \p(\xie^{+}) - \p(\zeta) \right|= \left| \int_{\mathbb{R}^{3}} \e^{-\im \zeta\cdot v} \left( \e^{-\im(\eta^{+} + \xi^{-} -\xie^{-})\cdot v} -1  \right) \rd F(v) \right|\\
	%\leq &  \int_{\mathbb{R}^{3}} \left|\e^{-\im \zeta\cdot v} \right|  \left( 1 - \e^{-\im(\eta^{+} + \xi^{-} -\xie^{-})\cdot v} \right) \rd F(v) \\
	%= & 1 - \p(\eta^{+} + \xi^{-} -\xie^{-})\\
	%\leq & \left\| 1-\p\right\|_{\alpha} \left| \eta^{+} + \xi^{-} -\xie^{-}\right|^{\alpha}\\
	%\leq & \left\| 1-\p\right\|_{\alpha} \left( \left| \eta^{+} \right|^{\alpha} + \left| \xi^{-} -\xie^{-}\right|^{\alpha} \right) \label{i1}\\
	%\leq & \left\| 1-\p\right\|_{\alpha} \left[\left( \left| \xi^{+}\right|^{\alpha}\sin^{\alpha}\left(\frac{\theta}{2}\right)\right) + \left(1-\ap\right)^{\alpha} \left| \xi^{-}\right|^{\alpha} \right] \label{i2}\\
	%\leq & \left\| 1-\p\right\|_{\alpha} \left[\left( \left| \xi\right|^{\alpha}\sin^{\alpha}\left(\frac{\theta}{2}\right)\right) + \left(1-\ap\right)^{\alpha} \left| \xi\right|^{\alpha}\sin^{\alpha}\left(\frac{\theta}{2}\right) \right] \label{i3}\\
	%\leq & C_{e} \left\| 1-\p\right\|_{\alpha} |\xi|^{\alpha} \sin^{\alpha}\left(\frac{\theta}{2}\right), 
	%and $ \xie^{-} = \ap\xi^{-} $ from \eqref{i1} to \eqref{i2} as well as the fact $ |\xi^{-}| = |\xi|\sin\left(\frac{\theta}{2}\right) $ from \eqref{i2} to \eqref{i3}
	where we utilize the relationship $ \left| \eta^{+}\right| = \left| \xi^{+}\right|\sin\left(\frac{\theta}{2}\right) $ and $ \left| \xi^{+}\right| \leq \left| \xi\right| $ in the last inequality. As a result, we have, according to the assumption \eqref{noncutoffb},
	\begin{equation}
	\left| I_{1}\right| \leq C_{1} \left\|1-\p\right\|_{\alpha} \left| \xi\right|^{\alpha} \int_{0}^{\frac{\pi}{2}} \sin^{\alpha}\left(\frac{\theta}{2}\right)b(\cos\theta) \sin\theta \rd \theta < \infty.
	\end{equation}
	(ii) For the second part $ I_{2} $, with the help of the inequality \eqref{l3} in Lemma \ref{L1} and $ \zeta_{e}-\xi = \eta_{e} $ in Figure \ref{fig:label}, we have 
	\begin{equation}
	\left| \p(\zeta_{e}) - \p(\xi)\right| \leq \left\| \p-1 \right\|_{\alpha} \left( 4 |\xi|^{\frac{\alpha}{2}} |\eta_{e}|^{\frac{\alpha}{2}} + |\eta_{e}|^{\alpha} \right),
	\end{equation}
	together with the geometric relation $|\zeta_{e}-\xi| = |\eta_{e}| \leq |\eta| = |\zeta - \xi| = |\xi|\sin^{2}\left(\frac{\theta}{2}\right) $, we can further obtain that
	\begin{align}
	\left| I_{2}\right| \leq & \int_{\Sd^{2}} b\left(\frac{\xi\cdot\sigma}{|\xi|}\right) \left\| \p-1 \right\|_{\alpha} \left( 4 |\xi|^{\frac{\alpha}{2}} |\eta_{e}|^{\frac{\alpha}{2}} + |\eta_{e}|^{\alpha} \right) \rd\sigma\\
	\leq & C_{2} \left\|1-\p\right\|_{\alpha} \left| \xi\right|^{\alpha} \int_{0}^{\frac{\pi}{2}} \sin^{\alpha}\left(\frac{\theta}{2}\right)b(\cos\theta) \sin\theta \rd \theta < \infty.
	\end{align}
	(iii) For the last part $ I_{3} $, following the similar estimates above and considering the fact that $ \left| \p(\xie^{+}) \right| \leq 1$, we have, 
	\begin{align}
	&\left|I_{3}\right| = \left| \int_{\Sd^{2}} b\left(\frac{\xi\cdot\sigma}{|\xi|}\right) \p(\xie^{+})\left[\p(\xie^{-}) - 1 \right] \rd\sigma \right|\\
	\leq & \int_{\Sd^{2}} b\left(\frac{\xi\cdot\sigma}{|\xi|}\right) \left|\p(\xie^{-}) - 1 \right| \rd\sigma\\
	\leq & \int_{\Sd^{2}} b\left(\frac{\xi\cdot\sigma}{|\xi|}\right) \left\|1-\p\right\|_{\alpha} \left| \xie^{-} \right|^{\alpha} \rd\sigma\\
	\leq & \left\|1-\p\right\|_{\alpha} \int_{\Sd^{2}} b\left(\frac{\xi\cdot\sigma}{|\xi|}\right)  \left(\frac{\ap^{2}}{2}\right)^{\frac{\alpha}{2}} \left( 1-\frac{\xi\cdot\sigma}{|\xi|}\right)^{\frac{\alpha}{2}} \left|\xi\right|^{\alpha} \rd\sigma\\
	\leq & C_{3}\left\|1-\p\right\|_{\alpha}\left|\xi\right|^{\alpha} \int_{0}^{\frac{\pi}{2}} \sin^{\alpha}\left(\frac{\theta}{2}\right)b(\cos\theta) \sin\theta \rd \theta,
	\end{align}
	where we use the estimate \eqref{bound-} in Lemma \ref{bound} as well as the fact that $ \frac{\xi\cdot\sigma}{|\xi|} = \cos\theta $. Summing up the estimates in (i), (ii) and (iii), we obtain the desired estimate \eqref{well2}.
\end{proof}

\begin{remark}
	In fact, without considering the geometric relation in Figure \ref{fig:label}, we can still find that the initial value problem \eqref{IBE}-\eqref{initial} is well-defined if there is only mild singularity assumption \eqref{mildcutoff}, by the following simple calculation,
	\begin{equation}
	\begin{split}
	\partial_{t} \varphi(t,\xi) = & \int_{\Sd^{2}} b\left(\frac{\xi\cdot\sigma}{|\xi|}\right) \left[ \varphi(t,\xie^{+})\varphi(t,\xie^{-}) - \varphi(t,\xi)\varphi(t,0) \right] \,\rd\sigma\\
	\leq & 4\int_{\Sd^{2}} b\left(\frac{\xi\cdot\sigma}{|\xi|}\right) \left|\xie^{+}\right|^{\frac{\alpha}{2}}\left|\xie^{-}\right|^{\frac{\alpha}{2}}\left\|\p-1\right\|_{\alpha}  \,\rd\sigma \\
	\leq & 4 \left(\ap^{2}\right)^{\frac{\alpha}{4}} \left[ \frac{\left(1+\am\right)^{2}+ \left(\ap\right)^{2}}{2}\right]^{\frac{\alpha}{4}} \left|\xi\right|^{\alpha} \left\| \p  -1\right\|_{\alpha}\\
	& \int_{\Sd^{2}} b\left(\frac{\xi\cdot\sigma}{|\xi|}\right) \left(\frac{1-\frac{\xi\cdot\sigma}{|\xi|}}{2}\right)^{\frac{\alpha}{4}}\left( \frac{1+\frac{\xi\cdot\sigma}{|\xi|}}{2}\right)^{\frac{\alpha}{4}} \rd\sigma \\
	= & 4\gamma_{\frac{\alpha}{2}} \left(\ap^{2}\right)^{\frac{\alpha}{4}} \left[ \frac{\left(1+\am\right)^{2}+ \left(\ap\right)^{2}}{2}\right]^{\frac{\alpha}{4}} \left|\xi\right|^{\alpha} \left\| \p  -1\right\|_{\alpha} < \infty,
	\end{split}
	\end{equation}
	where $ \gamma_{\frac{\alpha}{2}} $ has the same definition as in \eqref{gammaelastic} below, and we utilize the estimate \eqref{1} of Lemma \ref{well} in the first inequality as well as estimate \eqref{bound+}-\eqref{bound-} of Lemma \ref{bound} in the second inequality above.
\end{remark}

\section{Well-posedness under the Cutoff Assumption}
\label{sec:cutoff}

\subsection{Technical Lemma of the Cutoff Collision Operator}
\label{subsec:technical}

In this section, we first construct the solution of the initial value problem \eqref{IBE}-\eqref{initial}, and study its stability in the space $ \K^{\alpha} $ under the \textit{cutoff assumption} on the collision kernel $ b $ in the sense that, for all $ \xi\in \mathbb{R}^{3}/\{0\} $,
\begin{equation}\label{cutoff}
\int_{\Sd^{2}} b\left(\frac{\xi\cdot\sigma}{|\xi|}\right) \rd \sigma < \infty,
\end{equation}
in fact, we will dispense with the assumption and prove the existence of solutions to the initial value problem \eqref{IBE}-\eqref{initial} by compactness argument in next section \ref{sec:noncutoff}.

%\subsection{Construction of Solution by Banach Fixed Point Theorem }
%\label{subsec:banach}

Before that, we introduce some corresponded parameters that will appear systematically in our following proof. 
\begin{lemma}\label{parametergl}
	Assume that $ e\in (0,1] $ and the collision kernel $ b $ satisfy the cutoff assumption \eqref{cutoff}, for all $ \alpha\in\left[0,2\right] $ and $ \xi\in\mathbb{R}^{3}/\left\lbrace 0 \right\rbrace  $, we define the parameter $ \gamma_{e,\alpha} $,
	\begin{equation}\label{gammae}
	\gamma_{e,\alpha} \equiv \int_{\Sd^{2}} b\left(\frac{\xi\cdot\sigma}{|\xi|}\right) \frac{|\xie^{+}|^{\alpha}+|\xie^{-}|^{\alpha}}{|\xi|^{\alpha}} \rd \sigma,
	\end{equation}
	and $ \gamma_{e,\alpha} = \gamma_{\alpha}$ if and only if the restitution coefficient $ e=1 $, where the $ \gamma_{\alpha} $ is the corresponded parameter in elastic case, 
	\begin{equation}\label{gammaelastic}
	\gamma_{\alpha} \equiv \int_{\Sd^{2}} b\left(\frac{\xi\cdot\sigma}{|\xi|}\right) \frac{|\xi^{+}|^{\alpha}+|\xi^{-}|^{\alpha}}{|\xi|^{\alpha}} \rd \sigma = 2\pi \int_{0}^{\frac{\pi}{2}} b(\cos\theta) \left(\sin^{\alpha}\frac{\theta}{2} + \cos^{\alpha}\frac{\theta}{2}\right) \sin\theta\rd\theta.
	\end{equation}
	Furthermore, if the collision kernel $ b $ satisfy the non-cutoff assumption \eqref{noncutoffb}, we have the parameter $ \lambda_{e,\alpha} $ defined as \eqref{lamdadefinitionmain} above,
	\begin{equation}\label{lamdadefinition1}
	\lambda_{e,\alpha} \equiv \int_{\Sd^{2}} b\left(\frac{\xi\cdot\sigma}{|\xi|}\right) \left(\frac{|\xie^{+}|^{\alpha}+|\xie^{-}|^{\alpha}}{|\xi|^{\alpha}} -1 \right)\rd \sigma.
	\end{equation}
	Then $ \gamma_{e,\alpha} $ and $ \lambda_{e,\alpha} $ are finite and independent of $ |\xi| $.
\end{lemma}
\begin{proof}
	The proof is followed from the direct calculation by substituting the estimates of $ \xie^{+} $ and $ \xie^{-} $ in the Lemma \ref{bound}:
	for $ \gamma_{e,\alpha} $,
	\begin{equation}
	\left(\ap^{2}\right)^{\frac{\alpha}{2}} \left[\ap(1+\am)\right]^{\frac{\alpha}{2}} \gamma_{\alpha}  \leq \gamma_{e,\alpha} \leq \left(\ap^{2}\right)^{\frac{\alpha}{2}} \left[ \frac{\left(1+\am\right)^{2}+ \left(\ap\right)^{2}}{2}\right]^{\frac{\alpha}{2}} \gamma_{\alpha}.
	\end{equation}
	Note that the property of $ \gamma_{\alpha} $ has been proved in \cite[Lemma 4.1]{cannone2010infinite} corresponding to the elastic case. 
	
	For $ \lambda_{e,\alpha} $ under cutoff assumption, the finiteness can be immediately found with the help of $ \gamma_{e,\alpha} $ in \eqref{cutoff}; then to handle the non-cutoff collision kernel \eqref{noncutoffb}, we have the following estimate,
	\begin{equation}
	\begin{split}
	\lambda_{e,\alpha} = &  \int_{\Sd^{2}} b\left(\frac{\xi\cdot\sigma}{|\xi|}\right) \left( \frac{|\xie^{+}|^{\alpha}}{|\xi|^{\alpha}} + \frac{|\xie^{-}|^{\alpha}}{|\xi|^{\alpha}} -1 \right)\rd \sigma\\
	\leq &  \int_{\Sd^{2}} b\left(\frac{\xi\cdot\sigma}{|\xi|}\right) \left(\ap^{2}\right)^{\frac{\alpha}{2}} \left( \frac{1-\frac{\xi\cdot\sigma}{|\xi|}}{2}\right)^{\frac{\alpha}{2}} \rd \sigma\\
	= & 2\pi \left(\ap^{2}\right)^{\frac{\alpha}{2}} \int_{0}^{\frac{\pi}{2}} \sin^{\alpha}\left(\frac{\theta}{2}\right)b(\cos\theta) \sin\theta \rd \theta < \infty,
	\end{split}
	\end{equation}
	where in the middle inequality we apply the geometric relation $ |\xie^{+}| \leq |\xi| $ in Figure \ref{fig:label} as well as the estimate \eqref{bound-} of $ |\xie^{-}| $. This completes the proof of the Lemma \ref{parametergl}.
\end{proof}

In order to construct the solution by Banach fixed point theorem, we also present another technical Lemma \ref{GL} about the inelastic nonlinear operator $ \mathcal{G}_{e}[\p] $, defined as following:
\begin{equation}
\mathcal{G}_{e}[\p](\xi) := \int_{\Sd^{2}} b\left(\frac{\xi\cdot\sigma}{|\xi|}\right) \p(\xie^{+})\p(\xie^{-}) \ \rd \sigma,
\end{equation}
where $ \xie^{+} $ and $ \xie^{-} $ are defined in \eqref{xie+} and \eqref{xie-}.
\begin{lemma}\label{GL}
	Let $ e\in(0,1] $, $ \alpha\in \left[0,2\right] $ and the collision kernel $b$ satisfy the cutoff assumption \eqref{cutoff}. For any $\p\in\K^{\alpha}$, the function $ \mathcal{G}_{e}[\p] $ is continuous and positive definite. Moreover, we have 
	\begin{equation}\label{G}
	\left| \mathcal{G}_{e}[\p](\xi) - \mathcal{G}_{e}\left[ \tilde{\p} \right]  (\xi) \right| \leq \gamma_{e,\alpha} \left\| \p - \tilde{\p} \right\|_{\alpha} \left| \xi \right|^{\alpha}
	\end{equation}
	for all $ \p, \tilde{\p} \in \K^{\alpha} $ and all $ \xi\in\mathbb{R}^{3}/ \{ 0 \} $.
\end{lemma}
\begin{proof}
	For all $\p\in\K^{\alpha}$, to show that $ \mathcal{G}_{e}[\p] $ is continuous and positive definite, it suffices to show the estimate \eqref{G} holds, since the properties for $\p\in\K^{\alpha}$, we have $ \left| \p(\xie^{-}) \right| \leq 1 $, $ \left| \tilde{\p}(\xie^{+}) \right| \leq 1 $, we obtain
	\begin{equation}
	\begin{split}
	&\left| \mathcal{G}_{e}[\p](\xi) - \mathcal{G}_{e}\left[\tilde{\p} \right]  (\xi) \right| \\
	= & \left| \int_{\Sd^{2}} b\left(\frac{\xi\cdot\sigma}{|\xi|}\right) \left[ \left( \p(\xie^{+}) - \tilde{\p}(\xie^{+}) \right)\p(\xie^{-})  + \tilde{\p}(\xie^{+}) \left( \p(\xie^{-}) - \tilde{\p}(\xie^{-}) \right) \right] \ \rd \sigma  \right|\\
	\leq & \int_{\Sd^{2}} b\left(\frac{\xi\cdot\sigma}{|\xi|}\right) \left( \left\| \p-\tilde{\p} \right\|_{\alpha} \left| \xie^{+} \right|^{\alpha}  + \left\| \p-\tilde{\p} \right\|_{\alpha}\left| \xie^{-} \right|^{\alpha} \right) \ \rd \sigma\\
	= & \left\| \p-\tilde{\p} \right\|_{\alpha} \int_{\Sd^{2}} b\left(\frac{\xi\cdot\sigma}{|\xi|}\right) \left( \left| \xie^{+} \right|^{\alpha}  + \left| \xie^{-} \right|^{\alpha} \right) \ \rd \sigma\\
	\leq & \gamma_{e,\alpha} \left\| \p-\tilde{\p} \right\|_{\alpha} \left| \xi \right|^{\alpha}
	\end{split}
	\end{equation}
	for all $ \xi\in\mathbb{R}^{3} $.
\end{proof}

\subsection{Well-posedness under Cutoff Assumption}
\label{subsec:proofcutoff}

Now we are ready to give the construction of solution to the initial value equation \eqref{IBE}-\eqref{initial} in space $ \K^{\alpha} $.
Firstly, based on the cutoff assumption \eqref{cutoff}, we denote the consistent notation as in \cite{cannone2010infinite}, 
\begin{equation}
\gamma_{2} = \int_{\Sd^{2}} b\left(\frac{\xi\cdot\sigma}{|\xi|}\right) \rd \sigma = 2\pi \int_{0}^{\frac{\pi}{2}} b(\cos\theta) \sin\theta \rd \theta < \infty,
\end{equation}
meanwhile, considering the fact that $ \p(0,\xi) = 1 $ for all $ t \geq 0 $, we are able to rewrite the equation \eqref{IBE} into the following form:
\begin{equation}\label{i}
\partial_{t} \p(t,\xi) + \gamma_{2} \p(t,\xi) = \mathcal{G}_{e}[\p]  (t,\xi).
\end{equation}
Then multiplying \eqref{i} by the factor $ \e^{\gamma_{2}t} $ and integrating with respect to $ t $, we obtain the following equivalent formulation of equation \eqref{IBE}-\eqref{initial}:
\begin{equation}\label{mildsolution}
\p(t,\xi) = \p_{0}(\xi) \e^{-\gamma_{2}t} + \int_{0}^{t} \e^{-\gamma_{2}(t-\tau)} \mathcal{G}_{e}[\p] (\tau,\xi) \,\rd\tau.
\end{equation}

\begin{theorem}\label{local}
	\emph{(Well-posedness under cutoff assumption)} Let $ e\in(0,1] $, $ \alpha\in \left[0,2\right] $ and the collision kernel $b$ satisfy the cutoff assumption \eqref{cutoff}. For each initial datum $ \p_{0}\in\K^{\alpha} $, there exists a unique solution $ \p(t,\xi) $ to problem \eqref{IBE}-\eqref{initial} such that $ \p\in \chi^{\alpha}:= C\left(  \left[0,\infty\right), \K^{\alpha} \right) $.\\
	Furthermore, if $ \p, \tilde{\p} \in C\left( \left[0,\infty\right), \K^{\alpha} \right) $ are two solutions corresponding to the initial datum $ \p_{0}, \tilde{\p}_{0} $ respectively. Then, for every $ t \geq 0 $ and $ R\in\left(0,\infty\right] $,
	\begin{equation}\label{stabilitycutoff}
	\left\| \p(t,\cdot)-\tilde{\p}(t,\cdot) \right\|_{\alpha, R} \leq \e^{\lambda_{e,\alpha}t} \left\| \p_{0}-\tilde{\p}_{0} \right\|_{\alpha, R}
	\end{equation}
	in the sense of the quasi-metric as following: for any $ R\in\left( 0, \infty \right] $ and $ \p, \tilde{\p} \in \K^{\alpha} $,
	\begin{equation}
	\left\| \p(t,\cdot)-\tilde{\p}(t,\cdot) \right\|_{\alpha, R} \equiv \sup_{|\xi|\leq R} \frac{\left| \p(t,\xi) - \tilde{\p}(t,\xi) \right|}{|\xi|^{\alpha}}, 
	\end{equation}
	where the constant $ \lambda_{e,\alpha} = \gamma_{e,\alpha} - \gamma_{2} $. 
\end{theorem}
\begin{proof}
	\textit{(I) Proof of Existence and Uniqueness:}
	For fixed $ \p_{0}\in\K^{\alpha}$ and any $ \p\in\K^{\alpha} $, we're ready to apply the Banach fixed point theorem to the non-linear operator,
	\begin{equation}\label{P}
	\mathcal{P}[\p](t,\xi) \equiv \p_{0}(\xi) \e^{-\gamma_{2}t} + \int_{0}^{t} \e^{-\gamma_{2}(t-\tau)} \mathcal{G}_{e}[\p](\tau,\xi) \,\rd\tau.
	\end{equation}
	We prove the local existence and uniqueness by showing that operator $\mathcal{P} : \chi^{\alpha}_{T} \mapsto C( [0,T], \K^{\alpha} )$ has a unique fixed point in the space  $\chi^{\alpha}_{T} \subset C( [0,T], \K^{\alpha} )$ defined as
	\begin{equation}
	\chi^{\alpha}_{T} := \left\lbrace  \p\in C\left( [0,T],\K^{\alpha} \right) : \sup\limits_{t\in [0,T]}\left\|\p(t, \cdot)\right\|_{\alpha} < \infty \right\rbrace,
	\end{equation}
	which is a complete metric space with respect to the induced norm
	\begin{equation}
	\left\| \cdot \right\|_{\chi^{\alpha}_{T}}: =\sup\limits_{t\in [0,T]} \left\| \cdot \right\|_{\alpha}.
	\end{equation}
	(i) We need to show that, for any $ \p\in \chi^{\alpha}_{T} $ and every $ t\in\left[0,T\right] $, the function $ \mathcal{P}[\p](\cdot,\xi) \in \K^{\alpha}$, which means that $ \mathcal{P}[\p](t,\xi) $ is still continuous and positive definite: actually considering the Lemma \ref{GL}, we find that $ \mathcal{G}_{e}[\p](\tau,\cdot)$ is continuous and positive definite for every $ \tau\in\left[0,t\right] $, then $ \mathcal{P}[\p](t,\xi) \in \K^{\alpha}$ can directly follow the \cite[Lemma 3.5]{cannone2010infinite} (which implies that the linear combination with positive coefficients of positive definite functions is still a positive definite function), if one approximates the integral on the right hand side of \eqref{P} by finite sums with positive coefficients.
	
	Hence, for every $  \p\in \chi^{\alpha}_{T} $, by noticing the integration that $ \gamma_{2} \int_{0}^{t} \e^{-\gamma_{2}(t-\tau)} \rd \tau = 1 - \e^{-\gamma_{2}t} $, we rewrite the equation \eqref{P} as following 
	\begin{equation}
	\mathcal{P}[\p](t,\xi) - 1 = \left[\p_{0}(\xi)-1\right] \e^{-\gamma_{2}t} + \int_{0}^{t} \e^{-\gamma_{2}(t-\tau)} \left[\mathcal{G}_{e}[\p](\tau,\xi) - \gamma_{2}\right] \,\rd\tau.
	\end{equation}
	Furthermore, by the observation that $ \gamma_{2} = \mathcal{G}_{e}[1]$ as well as $ \e^{-\gamma_{2}(t-\tau)} \leq 1$ for every $ \tau\in\left[0,t\right] $, we obtain
	\begin{equation}
	\left| \mathcal{P}[\p](t,\xi) - 1 \right| \leq \left\| \p_{0} - 1 \right\|_{\alpha} \left| \xi \right|^{\alpha} + \gamma_{e,\alpha} \int_{0}^{t} \left\| \p(\tau,\xi) - 1 \right\|_{\alpha} \,\rd \tau \left| \xi \right|^{\alpha}.
	\end{equation} 
    After dividing the inequality above by $ |\xi|^{\alpha} $ and computing the supremum with respect to the variable $ \xi\in\mathbb{R}^{3} $ and $ t\in\left[0,T\right] $, we can finally prove that $ \mathcal{P}: \chi^{\alpha}_{T} \mapsto \chi^{\alpha}_{T} $ satisfying the following estimate:
    \begin{equation}
    \left\| \mathcal{P}[\p] - 1 \right\|_{\chi^{\alpha}_{T}} \leq \left\| \p_{0} - 1 \right\|_{\alpha} + \gamma_{e,\alpha} T \left\| \p - 1 \right\|_{\chi^{\alpha}_{T}} <\infty.
    \end{equation}
(ii) To prove that $ \mathcal{P}[\p](\cdot,\xi) \in \K^{\alpha}$ is a contraction in $\chi^{\alpha}_{T}$, we introduce another $ \mathcal{P}[\tilde{\p}](\cdot,\xi) \in \K^{\alpha} $, and make the subtraction between them. Then for the same initial datum $ \p_{0} $, we have, 
	\begin{equation}
	\begin{split}
	\left| \mathcal{P}[\p](t,\xi) - \mathcal{P}[\tilde{\p}](t,\xi) \right| \leq & \int_{0}^{t} \e^{-\gamma_{2}(t-\tau)} \left[\mathcal{G}_{e}[\p](\tau,\xi) - \mathcal{G}_{e}[\tilde{\p}](\tau,\xi) \right] \,\rd\tau\\
	\leq & \gamma_{e,\alpha} T \left\|\p - \tilde{\p}  \right\|_{\chi^{\alpha}_{T}} \left| \xi \right|^{\alpha}
	\end{split}
	\end{equation}
	where we utilize the Lemma \eqref{GL} in the last inequality. Consequently, after dividing the inequality above by $ |\xi|^{\alpha} $ with respect to the variable $ \xi\in\mathbb{R}^{3} $, we can obtain
	\begin{equation}
	\left\| \mathcal{P}[\p] - \mathcal{P}[\tilde{\p}] \right\|_{\chi^{\alpha}_{T}} \leq \gamma_{e,\alpha} T \left\|\p - \tilde{\p}  \right\|_{\chi^{\alpha}_{T}}.
	\end{equation}
	Combining (1) and (2), the Banach fixed point theorem provides the unique solution of \eqref{mildsolution} in the space $ \chi^{\alpha}_{T} $ provided that $ T < 1/\gamma_{e,\alpha} $.
	
	Note that finally we construct the unique solution on the time interval $ \left[ 0 ,T\right] $, where $ T $ is independent of the initial datum, therefore, by the continuation argument, we can extend the unique solution to $ \left[ T,2T\right] $ by choosing $ \p(T,\xi) $ as the initial datum. Consequently, repeating the same procedure, we manage to construct the unique solution on any finite time interval.
	
	\textit{(II) Proof of the Stability:}
	Starting from the function $ d\left(t,\xi\right) $ defined as following:
	\begin{equation}\label{d}
	d\left(t,\xi\right) := \frac{\p(t,\xi) - \tilde{\p}(t,\xi) }{|\xi|^{\alpha}},
	\end{equation}
	next, recalling equation \eqref{i} and the fact $ \p(t,0) = 1 $, we can obtain the equation satisfied by function $ d\left(t,\xi\right) $ after making subtraction between the equation \eqref{IBE} with respect to $ \p $ and $ \tilde{\p} $ separately:
	\begin{equation}
	\partial_{t} d\left(t,\xi\right) + \gamma_{2}d(t,\xi)= \int_{\Sd^{2}} b\left(\frac{\xi\cdot\sigma}{|\xi|}\right) \left[ \frac{\p(t,\xie^{+})\p(t,\xie^{-}) - \tilde{\p}(t,\xie^{+})\tilde{\p}(t,\xie^{-})}{|\xi|^{\alpha}} \right] \,\rd \sigma.
	\end{equation}
	Then, note that for $ |\xie^{+}| \leq R $ and $ |\xie^{-}| \leq R $, we have the inequality
	\begin{equation} \label{+2}
	\begin{split}
	&\left| \p(t,\xie^{+})\p(t,\xie^{-}) - \tilde{\p}(t,\xie^{+})\tilde{\p}(t,\xie^{-}) \right|\\
	\leq & \left| \p(t,\xie^{+})\p(t,\xie^{-}) - \tilde{\p}(t,\xie^{+})\p(t,\xie^{-}) + \tilde{\p}(t,\xie^{+})\p(t,\xie^{-}) - \tilde{\p}(t,\xie^{+})\tilde{\p}(t,\xie^{-}) \right|\\
	\leq & \left| \p(t,\xie^{+}) - \tilde{\p}(t,\xie^{+}) \right| \left| \p(t,\xie^{-}) \right| + \left| \p(t,\xie^{-}) - \tilde{\p}(t,\xie^{-}) \right| \left| \tilde{\p}(t,\xie^{+}) \right|\\
	\leq & \left\| \p(t,\cdot)-\tilde{\p}(t,\cdot) \right\|_{\alpha, R} \left(  \left| \xie^{+} \right|^{\alpha} + \left| \xie^{-} \right|^{\alpha} \right),
	\end{split}
	\end{equation}
	as a result, we further deduce the inequality satisfied by $ d(t,\xi) $, 
	\begin{equation}\label{s1}
	\partial_{t} d\left(t,\xi\right) + \gamma_{2}d(t,\xi) \leq \gamma_{e,\alpha} \left\| \p(t,\cdot)-\tilde{\p}(t,\cdot) \right\|_{\alpha, R}
	\end{equation}
	with the constants $ \gamma_{2} $ and $ \gamma_{e,\alpha} $. Moreover, we're able to solve the inequality \eqref{s1} by multiplying $ \e^{\gamma_{2}t} $ to both sides of it,
	\begin{equation}
	\partial_{t}\left( \e^{\gamma_{2}t} d(t,\xi) \right) \leq \gamma_{e,\alpha} \e^{\gamma_{2}t} \left\| \p(t,\cdot)-\tilde{\p}(t,\cdot) \right\|_{\alpha, R}
	\end{equation}
	and integrating the time variable from $ 0 $ to $ t $, hence,
	\begin{equation}
	\e^{\gamma_{2}t} d(t,\xi) \leq d(0,\xi) + \gamma_{e,\alpha}  \int_{0}^{t} \e^{\gamma_{2}s} \left\| \p(s,\cdot)-\tilde{\p}(s,\cdot) \right\|_{\alpha, R} \rd s.
	\end{equation}
	Finally, we compute the supremum with respect to $ |\xi| \leq R $,
	\begin{equation}
	\e^{\gamma_{2}t}\left\| \p(t,\cdot)-\tilde{\p}(t,\cdot) \right\|_{\alpha, R} \leq \left\| \p_{0}-\tilde{\p}_{0} \right\|_{\alpha, R} + \gamma_{e,\alpha} \int_{0}^{t} \e^{\gamma_{2}s} \left\| \p(s,\cdot)-\tilde{\p}(s,\cdot) \right\|_{\alpha, R} \rd s,
	\end{equation}
	and apply the integral form of Gr\" onwall's inequality to obtain
	\begin{equation}
	\left\| \p(t,\cdot)-\tilde{\p}(t,\cdot) \right\|_{\alpha, R} \leq \left\| \p_{0}-\tilde{\p}_{0} \right\|_{\alpha, R} \e^{\left(\gamma_{e,\alpha}-\gamma_{2}\right)t},
	\end{equation}
	where note that$ \gamma_{e,\alpha} - \gamma_{2} = \lambda_{e,\alpha} $ under cutoff assumption. In fact, though here the stability result \eqref{stabilitycutoff} is proved in the case of integrable collision kernel, but it can be generalized for the solutions to initial value problem \eqref{IBE}-\eqref{initial} with any non-cutoff collision kernel satisfying \eqref{noncutoffb} in the next section \ref{sec:noncutoff}.
\end{proof}

\section{Existence and Uniqueness with Non-Cutoff assumption}
\label{sec:noncutoff}
In this section, we complete the proof of the well-posedness of solutions to the initial value problem \eqref{IBE}-\eqref{initial} with non-cutoff assumption on the collision kernel, which implies that 
\begin{equation}\label{noncutoff}
\int_{\Sd^{2}} b\left(\frac{\xi\cdot\sigma}{|\xi|}\right) \ \rd \sigma = \infty,
\end{equation}
more precisely, $ b $ satisfies the singularity condition \eqref{noncutoffb}.

\subsection{Technical Lemma of the Non-Cutoff Collision Operator}
\label{subsec:technicalnoncutoff}
In fact, our strategy is to construct the solutions to \eqref{IBE}-\eqref{initial} with non-cutoff collision kernel based on compactness argument, hence, we first consider the increasing sequence of bounded collision kernels,
\begin{equation}
b_{n}(s) \equiv \min \left\lbrace b(s), n \right\rbrace \leq b(s), \quad n\in\mathbb{N},
\end{equation}
and, for every $ \alpha\in \left[ \alpha_{0}, 2 \right] $, the sequence of $ \p_{n} \in C\left( \left[0,\infty\right) ,\K^{\alpha} \right) $ of corresponding solutions to \eqref{IBE}-\eqref{initial} with cutoff collision kernels $ b_{n} $ and with the same initial datum $ \p_{0} \in \K^{\alpha} $. Furthermore, under the non-cutoff assumption \eqref{noncutoffb}, we have 
\begin{equation}
\lambda_{e,\alpha,n} \equiv \int_{\Sd^{2}} b_{n}\left(\frac{\xi\cdot\sigma}{|\xi|}\right) \left( \frac{|\xie^{+}|^{\alpha}+|\xie^{-}|^{\alpha}}{|\xi|^{\alpha}} -1 \right) \ \rd \sigma \leq \lambda_{e,\alpha},
\end{equation}
therefore, by the stability result \eqref{stabilitycutoff} with $ R = \infty $, it follows that
\begin{equation}\label{sta}
\left\| \p_{n}(t,\cdot)-1 \right\|_{\alpha} \leq \e^{\lambda_{e,\alpha,n}t} \left\| \p_{0}-1 \right\|_{\alpha} \leq \e^{\lambda_{e,\alpha}t} \left\| \p_{0}-1 \right\|_{\alpha},
\end{equation}
for all $ t\geq 0 $.

Before the specific proof of well-posedness theorem \ref{noncutoffwellposedness}, we give the following Lemma \eqref{uniequi} about the properties satisfied by the sequence of solution  $ \p_{n} \in C\left( \left[0,\infty\right) ,\K^{\alpha} \right) $,
\begin{lemma}\label{uniequi}
	Assume that $ e\in (0,1] $ and the collision kernel $ b $ satisfies the non-cutoff assumption \eqref{noncutoffb} with some $ \alpha_{0}\in\left[0,2\right] $. Let $ \alpha\in\left[\alpha_{0}, 2 \right] $, then the sequence of solutions $ \left\lbrace \p_{n}\right\rbrace^{\infty}_{n=1} \subset C\left(  \left[0,\infty\right), \K^{\alpha} \right) $ is bounded in $ C\left( \mathbb{R}^{3} \times \left[0,\infty\right) \right) $ and equicontinuous.
\end{lemma}
\begin{proof}
	\textit{Step I: Uniform Bound:} According to Theorem \ref{local}, the sequence of solution $ \p_{n}\left(t,\cdot\right) \in \K^{\alpha} $ under cutoff assumption are all chacteristic function for every $ t\geq 0 $, hence, we have 
	\begin{equation}
	\left| \p_{n}\left( t,\xi \right) \right| \leq \p_{n}(t,0) = 1,
	\end{equation}
	for all $ \xi\in\mathbb{R}^{3} $ and $ t\geq 0 $, which illustrates the uniform bound of $ \p_{n}\left( t,\xi \right) $.
	
	\textit{Step II: Continuity with respect to time variable $ t $.} We utilize the equation satisfied by $ \p_{n} $ as well as Lemma \ref{well} to obtain that
	\begin{equation}\label{step2}
	\begin{split}
	\left| \partial_{t} \p_{n}(t,\xi) \right| \leq & \int_{\Sd^{2}} b_{n}\left(\frac{\xi\cdot\sigma}{|\xi|}\right) \left| \varphi(t,\xie^{+})\varphi(t,\xie^{-}) - \varphi(t,\xi)\varphi(t,0) \right| \,\rd\sigma\\
	\leq & C_{e} \left\| \p_{n}(t,\cdot)- 1 \right\|_{\alpha} \left| \xi \right|^{\alpha} \left[ \int_{0}^{\frac{\pi}{2}} \sin^{\alpha}\left(\frac{\theta}{2}\right)b_{n}(\cos\theta) \sin\theta \,\rd \theta \right]\\
	\leq & C_{e} \e^{\lambda_{e,\alpha}t} \left\| \p_{0}-1 \right\|_{\alpha} \left| \xi \right|^{\alpha}\left[ \int_{0}^{\frac{\pi}{2}} \sin^{\alpha}\left(\frac{\theta}{2}\right)b_{n}(\cos\theta) \sin\theta \rd \theta \right],
	\end{split}
	\end{equation}
	for all $ \xi\in\mathbb{R}^{3} $ and $ t\geq 0 $, where we apply the stability result \eqref{sta}  in the last inequality. 
	
	\textit{Step III: Continuity with respect to fourier variable $ \xi $.} To prove this, it suffices to apply Lemma \ref{L1}, combined with Lemma \ref{reim} to obtain the following estimate:
	\begin{equation}
	\begin{split}
	\left| \p_{n} (t,\xi) - \p_{n}(t,\eta) \right| \leq & \sqrt{2\left[ 1- \Re\p_{n}\left(t,\xi-\eta\right)\right]}\\
	\leq & \sqrt{2}\left| \xi-\eta \right|^{\frac{\alpha}{2}} \left\| \p_{n}(t,\cdot) - 1 \right\|_{\alpha}^{\frac{1}{2}}\\
	\leq & \sqrt{2}\left| \xi-\eta \right|^{\frac{\alpha}{2}} \e^{\frac{\lambda_{e,\alpha}}{2}} \left\| \p_{0}-1 \right\|_{\alpha}^{\frac{1}{2}},
	\end{split}
	\end{equation}
	for all $ t\geq 0 $, where the stability result \eqref{sta} is used in the last inequality and the right-hand side is independent of $ n $.
\end{proof}

\subsection{Proof of Theorem \ref{noncutoffwellposedness}}
\label{subsec:proofnoncutoff}
Now in this subsection, we will present a complete proof of Theorem \ref{noncutoffwellposedness}, where the existence is guaranteed by the standard compactness argument and and uniqueness is given based on the stability estimate without cutoff assumption. 
\begin{proof}
	\textit{(I) Proof of  Existence:}
	According to the Ascoli-Arzel\` a theorem and the Cantor diagnal argument, we can deduce that there exists a subsequence of solutions $ \left\lbrace \p_{n_{k}}\right\rbrace_{n_{k}\in\mathbb{N}}  $ converging uniformly in any compact set of $ \mathbb{R}^{3}\times\left[0,\infty\right) $ based on the Lemma \ref{uniequi}. 
	
	Then, we need to prove the limit of functions $ \left\lbrace \p_{n_{k}}\right\rbrace_{n_{k}\in\mathbb{N}}  $,
	\begin{equation}\label{limit1}
	\p(t,\xi) = \lim\limits_{n_{k}\rightarrow\infty} \p_{n_{k}}(t,\xi)
	\end{equation}
	is the solution to the initial value problem \eqref{IBE}-\eqref{initial} under non-cutoff assumption \eqref{noncutoffb}. Note that $ \p\left(t,\cdot\right) $ is a characteristic function for every $ t\geq 0 $, as the pointwise limit of characteristic functions.\\
	Here we can apply the Lebesgue dominated convergence theorem to take the limit $ n_{k}\rightarrow\infty $ in the Boltzmann collision operator,
	\begin{equation}\label{B}
	\int_{\Sd^{2}} b_{n_{k}}\left(\frac{\xi\cdot\sigma}{|\xi|}\right) \left[ \p_{n_{k}}(t,\xie^{+})\p_{n_{k}}(t,\xie^{-}) - \p_{n_{k}}(t,\xi)\p_{n_{k}}(t,0) \right] \,\rd\sigma
	\end{equation}
	which, according to the calculation \eqref{step2} in the proof of Lemma \ref{uniequi}, can be controlled by the integrable function as following:
	\begin{equation}
	4 \e^{\lambda_{e,\alpha}t} \left\| \p_{0}-1 \right\|_{\alpha}  b\left(\frac{\xi\cdot\sigma}{|\xi|}\right) \left|\xie^{+}\right|^{\frac{\alpha}{2}}\left|\xie^{-}\right|^{\frac{\alpha}{2}}.
	\end{equation}
	On the other hand, since the Boltzmann collision operator in \eqref{B} converges uniformly on every compact subset of $ \mathbb{R}^{3} \times\left[0,\infty\right) $, there exists a continuous function $ \varsigma = \varsigma(t,\xi) $ such that $ \partial_{t}\p_{n_{k}} \rightarrow \varsigma $ as $ n_{k}\rightarrow\infty $. Meanwhile, considering the limit relation \eqref{limit1}, we immediately conclude that $ \varsigma = \partial_{t}\p $. Hence, the limit function $ \p(t,\xi) $ is a solution to the initial value problem \eqref{IBE}-\eqref{initial}. \\
	Finally, to show the limit function $ \p(\cdot,\xi) \in \K^{\alpha} $, it suffices to pass to the limit $ n_{k}\rightarrow\infty $  in the stability result \eqref{sta} in the following equivalent way
	\begin{equation}
	\frac{\left|  \p_{n_{k}}(t,\xi) - 1\right|}{|\xi|^{\alpha}} \leq \e^{\lambda_{e,\alpha}t} \left\| \p_{0}-1 \right\|_{\alpha} 
	\end{equation}
	for all $ \xi\in\mathbb{R}^{3} /\left\lbrace 0 \right\rbrace  $ and $ t\geq 0 $.
	
	\textit{(II) Proof of the Stability and Uniqueness:} As for the uniqueness of the solution we construct above, if we consider two sequences of solution $ \left\lbrace \p_{n}\right\rbrace_{n\in\mathbb{N}} $ and $ \left\lbrace \tilde{\p}_{n}\right\rbrace_{n\in\mathbb{N}} $ to the equation \eqref{IBE} with the cutoff kernel $ b_{n} $ as well as corresponding to the initial condition $ \p_{0} $ and $ \tilde{\p}_{0} $, respectively.\\
	By the compactness argument from Lemma \ref{uniequi}, there exists a subsequence $ n_{k}\rightarrow\infty $ and the solution to \eqref{IBE} by taking limit in the sense that
	\begin{equation}
	\p(t,\xi) = \lim\limits_{n_{k}\rightarrow\infty} \p_{n_{k}}(t,\xi) \quad \text{and} \quad \tilde{\p}(t,\xi) = \lim\limits_{n_{k}\rightarrow\infty} \tilde{\p}_{n_{k}}(t,\xi).
	\end{equation}
	Thus, in order to prove the uniqueness, we need to check the stability results under non-cutoff assumption: similar to the procedures under cutoff assumption, we have the following estimate by introducing the same $ d(t,\xi) $ as in \eqref{d} and dividing the integral domain of $ \sigma $ into four parts,
	\begin{equation}
	\begin{split}
	\partial_{t} d\left(t,\xi\right) = & \int_{\Sd^{2}} b\left(\frac{\xi\cdot\sigma}{|\xi|}\right) \left[ \frac{\p(t,\xie^{+})\p(t,\xie^{-}) - \tilde{\p}(t,\xie^{+})\tilde{\p}(t,\xie^{-})}{|\xi|^{\alpha}} - d(t,\xi) \right] \,\rd \sigma \\
	= & \int_{\Sd^{2}\cap \Omega^{c}_{\epsilon} } b\left(\frac{\xi\cdot\sigma}{|\xi|}\right) \left[ \frac{\p(t,\xie^{+})\p(t,\xie^{-}) - \tilde{\p}(t,\xie^{+})\tilde{\p}(t,\xie^{-})}{|\xi|^{\alpha}} \right]\rd \sigma - \left[\int_{\Sd^{2}\cap \Omega^{c}_{\epsilon} } b\left(\frac{\xi\cdot\sigma}{|\xi|}\right) \,\rd \sigma \right] d(t,\xi) \\
	+ & \int_{\Sd^{2}\cap \Omega_{\epsilon} } b\left(\frac{\xi\cdot\sigma}{|\xi|}\right) \left[ \frac{\p(t,\xie^{+})\p(t,\xie^{-}) - \tilde{\p}(t,0)\tilde{\p}(t,\xi)}{|\xi|^{\alpha}} \right]\,\rd \sigma \\
	- & \int_{\Sd^{2}\cap \Omega_{\epsilon} } b\left(\frac{\xi\cdot\sigma}{|\xi|}\right) \left[ \frac{\tilde{\p}(t,\xie^{+})\tilde{\p}(t,\xie^{-}) - \p(t,0)\p(t,\xi) }{|\xi|^{\alpha}} \right] \,\rd \sigma \\
	:= & I_{e,\epsilon} (t,\xi) - \gamma_{\epsilon}d(t,\xi) + R_{e,\p,\epsilon}(t,\xi) - R_{e,\tilde{\p},\epsilon}(t,\xi), 
	\end{split} % h change to d
	\end{equation}
	where $ \Omega_{\epsilon} $ ($ \Omega^{c}_{\epsilon} $ denotes its complement) is defined as 
	\begin{equation}
	\Omega_{\epsilon} := \Omega_{\epsilon}(\xi) = \left\lbrace \sigma\in\Sd^{2}; 1 - \frac{\xi}{|\xi|}\cdot\sigma \leq 2\left(\frac{\epsilon}{\pi}\right)^{2}\right\rbrace,
	\end{equation}
	for any $ \epsilon >0 $ and then $ \gamma_{\epsilon} $ can represented as 
	\begin{equation}
	\gamma_{\epsilon} = 2\pi \int_{\left[0,\frac{\pi}{2}\right]\cap \left\lbrace \sin\frac{\theta}{2} > \frac{\epsilon}{\pi} \right\rbrace } b(\cos\theta) \sin\theta \,\rd\theta \rightarrow \infty,
	\end{equation}
	as $ \epsilon\rightarrow 0^{+} $. Let $ R>0 $ and then with the help of \eqref{+2}, we have, for any $ |\xi| \leq R $,
	\begin{equation}
	\left|\frac{\p(t,\xie^{+})\p(t,\xie^{-}) - \tilde{\p}(t,\xie^{+})\tilde{\p}(t,\xie^{-})}{|\xi|^{\alpha}}\right| \leq \left\| \p(t,\cdot)-\tilde{\p}(t,\cdot) \right\|_{\alpha, R} \frac{\left| \xie^{+} \right|^{\alpha} + \left| \xie^{-} \right|^{\alpha}}{|\xi|^{\alpha}},
	\end{equation}%\left(\sup_{|\xi|\leq R}\left|d(\xi,t)\right|\right)
	combined the fact that $ \left| \xie^{\pm} \right| \leq \left|\xi\right| $, we further obtain,
	\begin{equation}
	\left|I_{e,\epsilon} (t,\xi)\right| \leq \gamma_{e,\alpha,\epsilon} \left\| \p(t,\cdot) -\tilde{\p}(t,\cdot) \right\|_{\alpha, R} \leq 2\gamma_{\alpha,\epsilon} \left\| \p(t,\cdot) -\tilde{\p}(t,\cdot) \right\|_{\alpha, R},
	\end{equation}
	where 
	\begin{equation}
	\gamma_{\alpha,\epsilon} = 2\pi \int_{\left[0,\frac{\pi}{2}\right]\cap \left\lbrace \sin\frac{\theta}{2} > \frac{\epsilon}{\pi} \right\rbrace } b(\cos\theta) \left(\sin^{\alpha}\frac{\theta}{2} + \cos^{\alpha}\frac{\theta}{2}\right) \sin\theta \,\rd\theta < \infty.
	\end{equation}
	Since the solutions $ \p(t,\xi), \tilde{\p}(t,\xi) \in C\left(  \left[0,\infty\right), \K^{\alpha} \right)$, it follows that for any fixed $ T>0 $,
	\begin{equation}
	\sup_{t\in(0,T],|\xi|\leq R} \left(\left|R_{e,\p,\epsilon}(t,\xi)\right| + \left|R_{e,\tilde{\p},\epsilon}(t,\xi)\right|\right) = r_{\epsilon} \rightarrow 0,
	\end{equation}
	as $ \epsilon\rightarrow 0^{+} $, which can be obtained by the following estimate with the help of Lemma \ref{wellsym},
	\begin{equation}
	\begin{split}
	\left|R_{e,\epsilon,\p}(t,\xi)\right| =& \left|\int_{\Sd^{2}\cap \Omega_{\epsilon}} b\left(\frac{\xi\cdot\sigma}{|\xi|}\right) \frac{\left[\p(t,\xie^{+})\p(t,\xie^{-}) - \p(t,\xi) \right]}{|\xi|^{\alpha}} \rd\sigma \right|\\
	\leq & C_{e} \left\|1-\p(t,\cdot) \right\|_{\alpha} \int_{0}^{\epsilon} \sin^{\alpha}\left(\frac{\theta}{2}\right)b(\cos\theta) \sin\theta \,\rd \theta \rightarrow 0
	\end{split}
	\end{equation}
	as $\epsilon\rightarrow 0^{+}$.\\
	Hence, we obtain the differential inequality of $ d(t,\xi) $, for any $ |\xi|\leq R $, 
	\begin{equation}
       \left|\partial_{t}d(t,\xi) + \gamma_{\epsilon}d(t,\xi)\right| \leq \gamma_{e,\alpha,\epsilon} \left\| \p(t,\cdot) -\tilde{\p}(t,\cdot) \right\|_{\alpha, R} + r_{\epsilon},
	\end{equation}
	and furthermore, by computing the supremum with respect to $ |\xi| \leq R $, we have 
	\begin{equation}
	\left\| \p(t,\cdot)-\tilde{\p}(t,\cdot) \right\|_{\alpha, R} \leq \e^{(\gamma_{e,\alpha,\epsilon}-\gamma_{\epsilon})t} \left\| \p_{0}-\tilde{\p}_{0}  \right\|_{\alpha, R} + \frac{r_{\epsilon}}{\gamma_{e,\alpha,\epsilon} - \gamma_{\epsilon}}\left[ \e^{(\gamma_{e,\alpha,\epsilon}-\gamma_{\epsilon})t} -1 \right].
	\end{equation}
	By taking the limit $ \epsilon\rightarrow 0 $ and letting $ R\rightarrow \infty $, we finally prove the stability result under non-cutoff assumption,
	\begin{equation}
	\left\| \p(t,\cdot)-\tilde{\p}(t,\cdot) \right\|_{\alpha} \leq \e^{\lambda_{e,\alpha}t} \left\| \p_{0}-\tilde{\p}_{0} \right\|_{\alpha},
	\end{equation}
	which then, implies the uniqueness of solution to \eqref{IBE}-\eqref{initial} in the space $ C \left( \left[0,\infty\right), \K^{\alpha} \right) $.

%	Furthermore, by using the stability results \eqref{sta} under the cutoff assumption, we obtain that
%	\begin{equation}
%	\frac{\left|  \p_{n_{k}}(\xi,t) - \tilde{\p}_{n_{k}}(\xi,t) \right|}{|\xi|^{\alpha}} \leq \e^{\lambda_{e,\alpha}t} \left\| \p_{0}-\tilde{\p}_{0} \right\|_{\alpha} 
%	\end{equation}
%	for all $ \xi\in\mathbb{R}^{3}/\left\lbrace 0 \right\rbrace  $ and $ t\geq 0 $. So far, by passing the limit $ {n_{k}} \rightarrow \infty $, we finally prove the stability results under non-cutoff assumption,
%	\begin{equation}
%	\left\| \p(t)-\tilde{\p}(t) \right\|_{\alpha} \leq \e^{\lambda_{e,\alpha}t} \left\| \p_{0}-\tilde{\p}_{0} \right\|_{\alpha}
%	\end{equation}
%	which then, implies the uniqueness of solution to \eqref{IBE}-\eqref{initial} in the space $ C \left( \left[0,\infty\right), \K^{\alpha} \right) $.
\end{proof}

\section{Large-time Asymptotic Behavior to Self-similar Solutions for the Inelastic Boltzmann Equation}
\label{sec:selfsimilar}
\subsection{Self-similar Solutions for the Inelastic Boltzmann Equation}
In this subsection, we will present the self-similar solution for the inelastic equation \eqref{IBE} in three-dimension, which may have infinite energy. Starting from introducing the isotropic function following the similar strategy as \cite{BC2002selfsimilarapplication}, 
\begin{equation}
u(t,x) = \p(t,|\xi|), \quad \text{where} \quad x = \frac{\left| \xi \right|^2}{2},
\end{equation}
together with the change of variable, we can reduce the original equation \eqref{IBE} to 
\begin{equation}\label{u1}
\partial_{t} u(t,x) = \int_{0}^{1} G\left(s\right) \left\lbrace u\left[t, \mathbf{a}(s) x \right] u\left[ t, \mathbf{b}(s) x \right] - u(t,0)u\left( t,x \right)\right\rbrace \,\rd s
\end{equation}
where 
\begin{equation}\label{asbs}
\mathbf{a}(s) = \ap^{2}s, \quad \mathbf{b}(s) = 1 - \ap\left(1+\am\right)s, \quad G(s) = \pi b\left(1-s\right)
\end{equation}
for any $ s\in\left(0,1\right) $, meanwhile, noting that
\begin{equation}
u(t,0) = \p(t,0) = \int_{\bR^{3}} f(v) \,\rd v = 1
\end{equation}
and the typical behaviour of characteristic functions of the infinite energy solution near the origin is described by the following asymptotic formula:
\begin{equation}
u(\cdot,x) = 1 - k x^{p} + O\left(x^{p+\epsilon}\right), \quad x\rightarrow 0^{+}, \quad 0<p = \frac{\alpha}{2}\leq 1, 
\end{equation}
with some $ k > 0 $ and $ \epsilon > 0 $. Considering the usual class of rapidly decreasing functions with $ p=1 $, we expect to extend this type of functions to real positive values of $ p $ by letting:
\begin{equation}\label{uform}
u(t,x) = \sum_{n=0}^{\infty} u_{n}(t) \frac{x^{np}}{\Gamma(np+1)}, \quad p > 0.
\end{equation}
In fact, such solutions for $ p > 1 $, which imply finite energy, have been considered in \cite{Bobylev1976invariant}, and then for $ 0 < p < 1 $, if one seeks for the solution in the form of \eqref{uform} and substitute the series of \eqref{uform} into equation \eqref{u1}, then the first two coefficients can be found immediately:
\begin{equation}
u_{0}(t) = 1, \quad u_{1}(t) = u_{1}(0) \e^{\lambda_{e,p}t} 
\end{equation}
where $ \lambda_{e,p} $ has the same form as \eqref{lamdadefinition1} after changing of variable,
\begin{equation}\label{lamdaep}
\lambda_{e,p} = \lambda_{e}(p) = \int_{0}^{1} G(s) \left[ \mathbf{a}(s)^{p} + \mathbf{b}(s)^{p} - 1\right] \,\rd s, \quad 0 <p <1.
\end{equation}
that is to say, the solution in the form of \eqref{uform} with $ 0 <p <1 $ has asymptotic behaviour for small enough $ x $ like:
\begin{equation}
u(t,x) \sim 1 - k x^{p} \e^{\lambda_{e,p}t} = 1 - k \left( x \e^{\mu_{e,p}t} \right)^{p}.
\end{equation}

Following the analysis above, we are now ready to state the next Proposition \ref{theoremx}, where the existence of the self-similar solution $ \Psi^{(p)}\left( x\e^{\mu_{e,p} t }\right) $ with respect to $ u(t,x) $ is presented; moreover, another special form solution $ \psi\left(t, x\e^{\mu_{e,p}t} \right) $ to \eqref{u1} with certain initial datum has been formulated as well, the limit of which is exactly the self-similar profile $ \Psi $.

\begin{proposition}\label{theoremx}
	Assume that $ e\in(0,1] $ and the scaled collision kernel $ 0\leq G(s) \leq k_{e}s^{-(1+\beta)} $ for some constants $ k_{e}>0 $ and $ 0 < \beta < 1 $, then for the initial condition as following,
	\begin{equation}\label{uinitial}
	\p_{0}(\xi)=u(0,x) = \sum_{n=0}^{\infty} u_{n}(0)\frac{x^{np}}{\Gamma\left(np+1\right)} \quad \text{with} \quad u_{0}(0) = 1, u_{1}(0) \neq 0,
	\end{equation} 
	where $ x = \left| \xi \right|^{2}/2 $ and $ p = \alpha/2 $, there exists a special unique solution $ u_{s}(t,x) $ to \eqref{u1} in the form
	\begin{equation}
	u_{s}(t,x) = \psi\left( t,x\e^{\mu_{e,p}t} \right),
	\end{equation}
	where $ \psi(t,x) $ is given by the series \eqref{psiform} with \eqref{psin} and $ \mu_{e,p} $ is defined as \eqref{muep} below.\\
	%	\begin{equation}\label{muep}
	%	\mu_{e,p} = \frac{1}{p} \int_{0}^{1} G(s) \left[ \mathbf{a}(s)^{p} + \mathbf{b}(s)^{p} - 1 \right] \rd s.
	%	\end{equation}
	Furthermore, for any constant $ \mu_{e,p} $ defined as \eqref{muep} above with $ \beta < p < 1 $, there exists a self-similar solution $ u(t,x) = \Psi^{(p)}\left( x\e^{\mu_{e,p} t }\right) $ given by the following series:
	\begin{equation}\label{Psiconclusion}
	\Psi^{(p)} (x) = \sum_{n=0}^{\infty} \Psi_{n}^{(p)} \frac{x^{np}}{\Gamma\left(np+1\right)} \quad \text{with} \quad \sup_{n = 1,2,3,...} \left| \Psi_{n}^{(p)} \right|^{\frac{1}{n}} < \infty,
	\end{equation}
	where $ \Psi_{0}^{(p)} = 1 $, $ \Psi_{1}^{(p)}\neq 0 $ can be chosen arbitrarily and $ \Psi_{n}^{(p)}(n = 2,3,... )$ are given by the recurrence formula as \eqref{oreccerence2}, such that 
	\begin{equation}
	\lim\limits_{t\rightarrow\infty} \psi(t,x) = \Psi^{(p)} (x),
	\end{equation}
	for any $ x \geq 0 $, provided $ \Psi_{1}^{(p)} = \psi_{1}(0) $.
\end{proposition}

\begin{proof}
For the sake of convenience, given $ u(t,x) = \p(t,|\xi|) $, we consider a new scaled function, for any $ 0 <p <1 $,
\begin{equation} \label{muep}
u \left(t,x\right) = \psi\left( t, x \e^{\mu_{e,p} t}\right), \quad \text{with} \quad \mu_{e,p} = \frac{\lambda_{e,p}}{p}
\end{equation}
which is apparently the solution to the following initial value problem,
\begin{equation}\label{psi1}
\partial_{t} \psi\left(t,x\right) + \mu_{e,p} x \cdot\nabla\psi = \int_{0}^{1} G\left(s\right) \left\lbrace \psi\left[\mathbf{a}(s) x \right] \psi\left[ \mathbf{b}(s) x \right] - \psi(0)\psi\left( x \right)\right\rbrace \,\rd s,
\end{equation}
with initial datum
\begin{equation}\label{psii1}
\psi\left(0,x\right) = u\left(0,x\right).
\end{equation}
Furthermore, in order to find the specific solution $ \psi $, we substitute the formal series 
\begin{equation}\label{psiform}
\psi(x) = \sum_{n=0}^{\infty} \psi_{n}(t) \frac{x^{np}}{\Gamma(np+1)}, \quad p > 0
\end{equation}
into the equation \eqref{psi1} and obtain the following set of recurrence equation:\
\begin{align}
\frac{\rd \psi_{0}}{\rd t} =& \frac{\rd \psi_{1}}{\rd t} = 0, \label{recurrence01}\\
\frac{\rd \psi_{n} }{\rd t} + \gamma_{e,n}(p) \psi_{n} =& \sum_{\begin{subarray}{c}i=1,\\ i+ j=n\end{subarray}}^{n-1} B_{e,p}(i,j) \psi_{i} \psi_{j} , \quad \text{for}\quad n = 2,3,... \label{recurrence}
\end{align}
where 
\begin{align}
\gamma_{e,n}(p) =& np\mu_{e,p} - \lambda_{e}(np) = n\lambda_{e,p} - \lambda_{e}(np), \\
%\lambda_{e,p} =& \lambda_{e}(p) =  \int_{0}^{1} G(s) \left[  \mathbf{a}(s)^{p} \mathbf{b}(s)^{p} - 1 \right] \ \rd s,\\
B_{e,p}(i,j) =& \frac{\Gamma(np+1)}{\Gamma(ip+1)\Gamma(jp+1)} \int_{0}^{1} G(s) \left[  \mathbf{a}(s)^{ip} \mathbf{b}(s)^{jp} \right] \ \rd s, \quad \text{for}\quad n = 2,3,...\label{Bij}
\end{align}
Moreover, thanks to the Leibniz integral rule, 
\begin{equation}
\lambda'_{e}(p) = \int_{0}^{1} G(s) \left[ \mathbf{a}(s)^{p}\ln \mathbf{a}(s) + \mathbf{b}(s)^{p}\ln \mathbf{b}(s) \right] \,\rd s,
\end{equation}
and considering the fact that $ 0 < \mathbf{a}(s) < 1$ and $ 0 < \mathbf{b}(s) < 1 $, we can further obtain $ \lambda'_{e}(p) < 0 $ and then the following estimate for $ \gamma_{e,n}(p) $,
\begin{equation}\label{gammaen1}
\gamma_{e,n}(p) = n\lambda_{e,p} - \lambda_{e}(np) \geq \left(n-1\right)\lambda_{e,p},
\end{equation}
such that $ \gamma_{e,n}(p) > 0 $, if $ n \geq 2 $. As a result, we are able to solve the recurrence relation \eqref{recurrence} of the coefficients $ \psi_{n}(t) $ that, for $ n = 2,3,... $,
\begin{equation}\label{psin}
\psi_{n}(t) = \psi_{n}(0) \e^{- \gamma_{e,n}(p) t} + \sum_{\begin{subarray}{c}i=1,\\ i+ j=n\end{subarray}}^{n-1} B_{e,p}(i,j) \int_{0}^{t} \e^{- \gamma_{e,n}(p) (t-\tau)} \psi_{i}(\tau) \psi_{j}(\tau) \ \rd\tau
\end{equation}
from which, we can formally deduce that, for $ n = 0,1,2,... $,
\begin{equation}
\psi_{n}(t) \rightarrow \Psi_{n}, \quad \text{as} \quad t\rightarrow\infty, 
\end{equation}
where $ \left\lbrace \Psi_{n} \right\rbrace_{n=0}^{\infty} $ are the steady solution to \eqref{recurrence01}-\eqref{recurrence} given by the recurrence relation:
\begin{align}
\Psi_{0} =& 1, \quad \Psi_{1} = \psi_{1}, \label{oreccerence1}\\
\Psi_{n} =& \frac{1}{\gamma_{e,n}(p)} \sum_{\begin{subarray}{c}i=1,\\ i+ j=n\end{subarray}}^{n-1} B_{e,p}(i,j) \Psi_{i} \Psi_{j}, \label{oreccerence2}
\end{align}
and are also the coefficients of the series solution $ \Psi $ to the following equation,
\begin{equation}\label{omega}
\mu x \cdot \nabla \Psi = \int_{0}^{1} G\left(s\right) \left\lbrace \Psi\left[\mathbf{a}(s) x \right] \Psi\left[ \mathbf{b}(s) x \right] - \Psi\left( x \right)\Psi\left(0\right)\right\rbrace \,\rd s,
\end{equation}
which is the corresponded steady equation derived by substituting self-similar profile $ u(t,x) = \Psi(x\e^{\mu_{e,p} t}) $ into \eqref{u1}.

As we mentioned before, so far our calculations above have been quite formal, as there is no evidence to show the convergence of series \eqref{psiform}, as a result, we are now prepared to rigorously prove the convergence of series \eqref{psiform}, by showing that the solutions $ \psi_{n}(t) $ have the following uniform bound $ A_{e}^{n} $, for any $ t\in\left[0,\infty\right) $,
\begin{equation}
\left| \psi_{n}(t) \right| \leq A_{e}^{n}, \quad \text{for} \quad n = 1,2,...
\end{equation}
under the assumption about the initial datum $ \psi_{n}(0) $ in the sense that there exists a constant $ A_{0}>0 $ such that
\begin{equation}\label{psiassumption1}
\left| \psi_{n}(0) \right| \leq A_{0}^{n}, \quad \text{for} \quad n = 1,2,...
\end{equation}
which suffices to guarantee the convergence of series of \eqref{psiform}. Thus, we can complete the proof combining with the following Lemma \ref{technical1}.
\end{proof}

Finally, in order to illustrate this, we present the technical Lemma \ref{technical1}, which will play an important role in proving the convergence of series \eqref{psiform} for the \textit{non-cutoff Maxwellian collision kernels}.
\begin{lemma}\label{technical1}
	Assume that $ e\in(0,1] $ and $ 0\leq G(s) \leq k_{e}s^{-(1+\beta)} $ for some constants $ k_{e}>0 $ and $ 0 < \beta < 1 $, then there exists a constant $ C=C(p,\beta) $ such that, for any $ p>\beta $, 
	\begin{equation}\label{n-1}
	\frac{1}{n-1} \sum_{\begin{subarray}{c}i=1,\\ i+ j=n\end{subarray}}^{n-1}  B_{e,p}(i,j)  \leq k_{e} C(p,\beta), \quad \text{for} \quad n= 2, 3,...
	\end{equation}
	where the definition of coefficients $ B_{e,p}(i,j) $ has been given in \eqref{Bij}. \\
	Furthermore, if the initial datums $ \psi_{n}(0) $ satisfy the assumption \eqref{psiassumption1}, then for any $ t\geq 0 $, 
	\begin{equation}\label{psiuniform}
	\left| \psi_{n}(t) \right| \leq A_{0}^{n} \left[ 1 + \frac{k_{e}}{\lambda_{e,p}} C(p,\beta)\right]^{n-1} \quad \text{for} \quad n = 1,2,... 
	\end{equation}
	where the definition of $ \lambda_{e,p} $ has been given in \eqref{lamdaep}.
\end{lemma}
\begin{proof}
	In \cite{BC2002selfsimilarapplication}, the similar Lemma is true for the elastic case, whose proof is based on the well-known identities of the classical Beta- and Gamma-functions, 
	\begin{equation}\label{GammaBeta}
	\int_{0}^{1} s^{y_{1}-1} \left(1-s\right)^{y_{2}-1} \rd s = \frac{\Gamma(y_{1})\Gamma(y_{2})}{\Gamma(y_{1}+y_{2})} \quad\text{and}\quad \lim\limits_{z\rightarrow\infty} \frac{\Gamma(z)z^{p}}{\Gamma(z+p)} = 1.
	\end{equation}
	Here we will extend the result to the inelastic case whenever restitution coefficient $ 0 < e \leq 1 $ with the help of some additional estimates. \\
	(i) By noticing that $ 1 + \am > \ap $ and formula \eqref{asbs}, we have
	\begin{equation}
	\mathbf{b}(s) = 1-\ap(1+\am)s \leq 1 - \ap^{2}s = 1 - \mathbf{a}(s)
	\end{equation}
	and then the formula \eqref{Bij} of coefficient $ B_{e,p}(i,j) $ has the following estimate with the help of first identity in \eqref{GammaBeta} as well as the assumption of $ G(s) $,
	\begin{equation}
	\begin{split}
	B_{e,p}(i,j) \leq& \frac{\Gamma(np+1)}{\Gamma(ip+1)\Gamma(jp+1)} \int_{0}^{1} ks^{-(1+\beta)} \left[  \mathbf{a}(s)^{ip} \left[1 - \mathbf{a}(s)\right]^{jp} \right] \rd s \\
	\leq& \frac{k}{\ap^{4}} \frac{\Gamma(np+1)}{\Gamma(ip+1)\Gamma(jp+1)} \int_{0}^{1} \mathbf{a}(s)^{ip-\beta-1}\left[1 - \mathbf{a}(s)\right]^{jp+1-1} \rd \mathbf{a}(s)\\
	\leq& k_{e} \frac{\Gamma(ip-\beta)\Gamma(np+1)}{\Gamma(ip+1)\Gamma(np+1-\beta)}
	\end{split}
	\end{equation}
	consequently, by summing up with respect to $ i $ and $ j $,
	\begin{equation}
	\frac{1}{n-1} \sum_{\begin{subarray}{c}i=1,\\ i+ j=n\end{subarray}}^{n-1}  B_{e,p}(i,j) \leq k_{e} \frac{\Gamma(np+1)}{(n-1)\Gamma(np+1-\beta)} \sum_{\begin{subarray}{c}i=1,\\ i+ j=n\end{subarray}}^{n-1} \frac{\Gamma(ip-\beta)}{\Gamma(ip+1)}, \quad \text{for} \quad n=2,3,...
	\end{equation}
	Thanks to the second identity in \eqref{GammaBeta}, we have, 
	\begin{equation}
	\lim\limits_{i\rightarrow\infty} \frac{\Gamma(ip-\beta)}{\Gamma(ip+1)} =  \left(ip\right)^{-(1+\beta)} \quad \text{and} \quad \lim\limits_{n\rightarrow\infty}\frac{\Gamma(np+1)}{\Gamma(np+1-\beta)} = \left(np\right)^{\beta}.
	\end{equation}
    from which, we can conclude that, for $ 0 < \beta < p < 1 $,
    \begin{equation}
    S(p,\beta) = \sum_{\begin{subarray}{c}i=1,\\ i+ j=n\end{subarray}}^{n-1} \frac{\Gamma(ip-\beta)}{\Gamma(ip+1)} < \infty \quad \text{and} \quad r(p,\beta) = \sup_{n = 2,3,...}\frac{\Gamma(np+1)}{(n-1)\Gamma(np+1-\beta)} < \infty,
    \end{equation}
    hence, we can obtain the estimate \eqref{n-1} by letting $ C(p,\beta) = r(p,\beta)S(p,\beta) $.\\
    (ii) As for the estimate \eqref{psiuniform}, we complete the proof by using the induction method: first of all, it is true for $ n=1 $ according to the recurrence formula \eqref{recurrence01}:
    \begin{equation}
    \psi_{1}(t) = \psi_{1}(0),
    \end{equation} 
    then we assume that the estimate \eqref{psiuniform} holds for $ n=2,3,...,m-1 $ with $ m-1\geq 2 $, and substitute the case $ n=m-1 $ into \eqref{psin} to obtain the estimate for $ n=m $ as following,
    \begin{equation}
    \left| \psi_{m}(t) \right| \leq A_{0}^{m} \left[ \e^{-\gamma_{e,m}(p)t} + b_{e}^{m-2} \sum_{\begin{subarray}{c}i=1,\\ i + j=m\end{subarray}}^{m-1}  B_{e,p}(i,j) \frac{1-\e^{-\gamma_{e,m}(p)t}}{\gamma_{e,m}(p)} \right],
    \end{equation}
    where 
    \begin{equation}\label{be}
    b_{e} = 1 + \frac{k_{e}}{\lambda_{e,p}} C(p,\beta), \quad \text{for} \quad  0 < \beta < p < 1.
    \end{equation}
    Meanwhile, note that the inequality \eqref{gammaen1} of $ \gamma_{e,m}(p) $ implies the fact that $ \e^{-\gamma_{e,m}(p)t} \leq 1 $ for any $ t \geq 0 $, which further results in the following estimate of $ \left| \psi_{m}(t) \right| $,
    \begin{equation}
    \left| \psi_{m}(t) \right| \leq A_{0}^{m} \left[ 1 + \frac{b_{e}^{m-2}}{\left(m-1\right)\lambda_{e,p}} \sum_{\begin{subarray}{c}i=1,\\ i + j=m\end{subarray}}^{m-1}  B_{e,p}(i,j) \right].
    \end{equation}
    Hence, according to the estimate \eqref{n-1} as well as the definition of $ b_{e} $ of \eqref{be}, we can obtain the final estimate of $ \left| \psi_{m}(t) \right| $,
    \begin{equation}
    \begin{split}
    \left| \psi_{m}(t) \right| \leq A_{0}^{m} \left[ 1 + \frac{b_{e}^{m-2}}{\lambda_{e,p}} k_{e} C\left(p,\beta\right) \right] = A_{0}^{m} \left[ 1 + b_{e}^{m-2} \left(b_{e} -1 \right)\right] \leq  A_{0}^{m} b_{e}^{m-1},
    \end{split}
    \end{equation}
    where we utilize the fact that $ b_{e} > 1 $ in the last inequality above. This completes the standard induction procedures.
\end{proof}
\begin{remark}
	By observing the recurrence relation \eqref{oreccerence1}-\eqref{oreccerence2}, the similar estimates can be obtained for coefficients $ \left\lbrace \Psi^{(p)}_{n} \right\rbrace_{n=0}^{\infty} $ that 
	\begin{equation}
	\left| \Psi^{(p)}_{n}\right| \leq \left| \Psi^{(p)}_{1} \right|^{n} \left(b_{e} -1 \right)^{n-1},
	\end{equation}
	where $ b_{e} $ is defined as \eqref{be}. 
\end{remark}

\subsection{Proof of the Theorem \ref{SteadyExistence}}
%Changing back to our original notation, we obtain that
%\begin{equation}
%\phi(\xi,t) = \psi(\xi\e^{\mu_{e,\alpha} t })
%\end{equation}

In this subsection, we give a detailed proof of Theorem \ref{SteadyExistence} about the existence of  steady solution $ \Phi $, which, in fact, is the direct consequence of Proposition \ref{theoremx} by changing variable $ x $ back to the original notation $ \eta $.

\begin{proof}
	For the singularity condition of the collision kernel, although the Lemma \ref{technical1} and Proposition \ref{theoremx} is proved under the assumption of the scaled collision kernel $ 0\leq G(s) \leq k_{e}s^{-(1+\beta)} $ for some constants $ k_{e}>0 $ and $ 0 < \beta < 1 $, this can be replaced by the assumption of original collision kernel form $ b $ with the help of the transformation $ G(s) = \pi b\left(1-s\right) $ in \eqref{asbs}. Indeed, after changing variables, $ [s(1-s)]^{\beta} G(s)\in L^{1}[0,1)$ will return to the assumption \eqref{noncutoffs} of $ b $, where the singularity appears at $ s\rightarrow 1 $, by setting $ \beta = \alpha/2 $:
	\begin{equation}\label{weaker}
	(1-s)^{\frac{\alpha}{2}} b(s) \in L^{1}[0,1), 
	\end{equation}
	for some $ \alpha\in [0,2] $, which actually can fall into our original non-cutoff assumption \eqref{noncutoffb}.
	% apparently satisfies the setting \eqref{weaker}.
	
	On the other hand, the steady solution $ \Phi(\eta) = \Phi(\left| \eta \right|) $ is constructed in the following form of the series by returning back $ \alpha = 2p $,
	\begin{equation}
	\Phi\left( \xi\e^{\mu_{e,\alpha} t }\right) = \Phi^{(\alpha)}_{e,K} \left(\eta\right)= \sum_{n=0}^{\infty}\Psi^{(\alpha)}_{n} \frac{ \left(\left| \eta \right|^{\alpha}\right)^{n}}{\Gamma\left(n \frac{\alpha}{2}+1\right)},
	\end{equation}
	which leads to the estimate \eqref{Phiasym}. Still, we need to prove the solution $ \Phi^{(\alpha)}_{e,K} $ is a characteristic function: in fact, we can conclude this by considering fact, if the initial datum is characteristic function, that the series \eqref{psiform} converges uniformly on $ t\in \left[0,\infty\right) $ to corresponded solution $ \psi $, which is a characteristic function at any $ t>0 $ by Lemma \ref{technical1}, and on the other hand, the $ \Psi $ is a pointwise limit of $ \psi $ as $ t\rightarrow \infty $ with uniqueness property. Thus, by changing back to variable $ \eta $, the steady solution $ \Phi^{(\alpha)}_{e,K} $ is also a characteristic function such that $ \Phi^{(\alpha)}_{e,K} \in \mathcal{K}^{\alpha} $.
\end{proof}

\subsection{Proof of the Asymptotic Stability Theorem \ref{AsymptoticStability}}
\label{subsec:asymptotic}
Finally we are in a position to give a complete proof of stability result of the rescaled initial value problem \eqref{phi}-\eqref{phii}, combined which, we can find that the solution $ \p(t,\xi) = \phi_{e}^{(\alpha)}(\xi\e^{\mu_{e,\alpha}t}, t) $ to \eqref{IBE}-\eqref{initial} converges (in self-similar variables) towards the self-similar profile $ \Phi $ under some specific initial condition.

\begin{proof}
The proof is partially relied on the stability result of $ \p(t,\xi) $, where it follows the stability results \eqref{stabilitynoncutoff} for any collision kernel satisfying the \eqref{noncutoffb}:
For any two solutions $ \p(t,\xi) = \phi_{e}^{(\alpha)}(t, \xi\e^{\mu_{e,\alpha}t}) $ and $ \tilde{\p}_{e}^{(\alpha)}(t,\xi) = \tilde{\phi}(t,\xi\e^{\mu_{e,\alpha}t}) $, by means of the observation under change of variable,
\begin{equation}\label{phistaa}
\sup_{|\xi|\leq R} \frac{\left|\phi_{e}^{(\alpha)}(t,\xi\e^{\mu_{e,\alpha}t}) - \tilde{\phi}_{e}^{(\alpha)}(t,\xi\e^{\mu_{e,\alpha}t}) \right|}{\left| \xi \right|^{\alpha}} = \e^{\lambda_{e,\alpha}t} \sup_{|\xi|\leq R\e^{\mu_{e,\alpha}t}} \frac{\left| \phi_{e}^{(\alpha)}(t,\xi) - \tilde{\phi}_{e}^{(\alpha)}(t,\xi) \right|}{\left| \xi \right|^{\alpha}}
\end{equation}
combined with \eqref{stabilitynoncutoff} such that, for all $ t>0 $ and $ R \in \left(0,\infty\right] $,
\begin{equation}\label{phistab}
\sup_{|\xi|\leq R} \frac{\left|\phi_{e}^{(\alpha)}(t,\xi\e^{\mu_{e,\alpha}t}) - \tilde{\phi}_{e}^{(\alpha)}(t,\xi\e^{\mu_{e,\alpha}t}) \right|}{\left| \xi \right|^{\alpha}} \leq \e^{\lambda_{e,\alpha}t} \sup_{|\xi|\leq R} \frac{\left| \phi_{0}(\xi) - \tilde{\phi}_{0}(\xi) \right|}{\left| \xi \right|^{\alpha}}
\end{equation}
we then obtain the estimate as following by linking \eqref{phistaa} with \eqref{phistab} ,
\begin{equation}
\sup_{|\xi|\leq R\e^{\mu_{e,\alpha}t}} \frac{\left| \phi_{e}^{(\alpha)}(t,\xi) - \tilde{\phi}_{e}^{(\alpha)}(t,\xi) \right|}{\left| \xi \right|^{\alpha}} \leq \sup_{|\xi|\leq R} \frac{\left| \phi_{0}(\xi) - \tilde{\phi}_{0}(\xi) \right|}{\left| \xi \right|^{\alpha}}.
\end{equation}
Moreover, let $ S = R\e^{\mu_{e,\alpha}t} $, we have
\begin{equation}\label{<S}
\sup_{|\xi|\leq S} \frac{\left| \phi_{e}^{(\alpha)}(t,\xi) - \tilde{\phi}_{e}^{(\alpha)}(t,\xi) \right|}{\left| \xi \right|^{\alpha}} \leq \sup_{|\xi|\leq S\e^{-\mu_{e,\alpha}t}} \frac{\left| \phi_{0}(\xi) - \tilde{\phi}_{0}(\xi) \right|}{\left| \xi \right|^{\alpha}}.
\end{equation}
Now we're able to complete the proof by study the estimate of $ \left\| \phi_{e}^{(\alpha)}(t,\cdot) - \tilde{\phi}_{e}^{(\alpha)}(t,\cdot) \right\|_{\alpha} $ as following
\begin{align}\label{stabilityIII}
\left\| \phi_{e}^{(\alpha)}(t,\cdot) - \tilde{\phi}_{e}^{(\alpha)}(t,\cdot) \right\|_{\alpha} =& \sup_{|\xi|\leq S} \frac{\left| \phi_{e}^{(\alpha)}(t,\xi) - \tilde{\phi}_{e}^{(\alpha)}(t,\xi) \right|}{\left| \xi \right|^{\alpha}} + \sup_{|\xi| > S} \frac{\left| \phi_{e}^{(\alpha)}(t,\xi) - \tilde{\phi}_{e}^{(\alpha)}(t,\xi) \right|}{\left| \xi \right|^{\alpha}}\\
:=& I_{1} + I_{2},
\end{align}
where we can get the estimate for $ I_{1} $ directly from \eqref{<S}. As for term $ I_{2} $, by recalling the fact that $ \left| \phi_{e}^{(\alpha)}(t,\xi) \right| < 1 $ and $ \left| \tilde{\phi}_{e}^{(\alpha)}(t,\xi) \right| < 1 $, we find that, for any arbitrary small $ \epsilon > 0 $, there exists $ S>0 $ such that
\begin{equation}
\sup_{|\xi| > S} \frac{\left| \phi_{e}^{(\alpha)}(t,\xi) - \tilde{\phi}_{e}^{(\alpha)}(t,\xi) \right|}{\left| \xi \right|^{\alpha}} \leq \sup_{|\xi| > S} \frac{2}{\left| \xi \right|^{\alpha}} \leq \frac{2}{R^{\alpha}} \leq \epsilon,
\end{equation}
where in the last two inequalities above we utilize that $ R = S \e^{-\mu_{e,\alpha}t} < S$, for all $ t>0 $ and each $ R \in \left(0,\infty\right] $.

Consequently, the estimate \eqref{stabilityIII} leads to that,
\begin{equation}\label{key}
\left\| \phi_{e}^{(\alpha)}(t,\cdot) - \tilde{\phi}_{e}^{(\alpha)}(t,\cdot) \right\|_{\alpha} \leq \sup_{|\xi|\leq S\e^{-\mu_{e,\alpha}t}} \frac{\left| \phi_{0}(\xi) - \tilde{\phi}_{0}(\xi) \right|}{\left| \xi \right|^{\alpha}} + \epsilon,
\end{equation}
and we can further conclude the large-time asymptotic stability by letting $ t\rightarrow \infty $ as well as noting the fact that $ \epsilon > 0 $ can be arbitrary small.
\end{proof}
%%%%%%%%%%%%%%%%%%%%%%%%%%%%%%%%%%%%%%%%%%%%%%%%%%%%%%%%%%%%%%%%%%%%%%%%%%%%%%%%%%%%%%%%%%%%%%%%%

%\listoftodos

\appendix
\section{Appendix}
\label{sec:app}

\subsection*{A.1 Fourier Transform of $ Q_{e}^{+} $}
\label{sub:fourier+}
For the sake of completeness, we present the Fourier transformation for the inelastic collision operator, where we try to keep consistency with the notation used in \cite[Theorem 12]{DesvillettesNote2003}. 
In the elastic case, after the Fourier transformation, we can get the beautiful formula, which is called Bobylev identity, likewise, we expect to find the formula of inelastic Boltzmann equation.
Here, we take the inelastic gain term $ Q^{+}_{e}(g,f)(v) $ as example, as the loss term $ \mathcal{F}\left[ Q^{-}_{e}(g,f) \right] $ is the same as the elastic case $ \mathcal{F}\left[ Q^{-}(g,f) \right] $. By performing the weak formulation, for any test function $ \phi $, we have,
\begin{equation}
\int_{\bR^{3}} Q^{+}_{e}(g,f)(v)\phi(v) \rd v =  \int_{\bR^{3}} \int_{\bR^{3}} \int_{\Sd^{2}} b\left(\frac{v-\vs}{|v-\vs|}\cdot \sigma\right) g(\vs) f(v) \phi(v')  \,\rd \sigma \,\rd \vs \,\rd v.
\end{equation}
Selecting $ \phi(v) = \e^{-iv\cdot\xi} $ in the identity above, we have
\begin{equation}
\begin{split}
&\mathcal{F}\left[ Q^{+}_{e}(g,f) \right](\xi) \\
=& \int_{\bR^{3}} \int_{\bR^{3}} \int_{\Sd^{2}} g(\vs) f(v) b\left(\frac{v-\vs}{|v-\vs|}\cdot \sigma\right) \e^{ -i \left(\frac{v+\vs}{2} + \frac{1-e}{4}(v-\vs) + \frac{1+e}{4}|v-\vs|\sigma\right)\cdot \xi}  \,\rd \sigma \,\rd \vs \,\rd v\\
=& \int_{\bR^{3}} \int_{\bR^{3}} \int_{\Sd^{2}} g(\vs) f(v) b\left(\frac{v-\vs}{|v-\vs|}\cdot \sigma\right) \e^{ -i\frac{v+\vs}{2}\cdot\xi}\e^{-i \left(\frac{1-e}{4}(v-\vs) + \frac{1+e}{4}|v-\vs|\sigma\right)\cdot\xi} \,\rd \sigma \,\rd \vs \,\rd v 
\end{split}
\end{equation}
according to the general change of variable,
\begin{equation}
\int_{\Sd^{2}} F(k\cdot\sigma, l\cdot\sigma) \rd \sigma = \int_{\Sd^{2}} F(l\cdot\sigma, k\cdot\sigma) \rd\sigma , \quad  |l| = |k| = 1,
\end{equation}
due to the existence of an isometry on $ \Sd^{2} $ exchanging $ l $ and $ k $, we have, by exchanging the rule of $ \frac{\xi}{|\xi|} $ and $ \frac{v-\vs}{|v-\vs|} $,
\begin{equation}
\begin{split}
\int_{\Sd^{2}} g(\vs) f(v) b\left(\frac{v-\vs}{|v-\vs|}\cdot \sigma\right) \e^{-i \left(\frac{1-e}{4}(v-\vs) + \frac{1+e}{4}|v-\vs|\sigma\right)\cdot\xi}  \,\rd \sigma \\
= \int_{\Sd^{2}} g(\vs) f(v) b\left(\frac{\xi}{|\xi|}\cdot \sigma\right) \e^{-i \left(\frac{1-e}{4}\xi + \frac{1+e}{4}|\xi|\sigma\right)\cdot(v-\vs)}  \,\rd \sigma
\end{split}
\end{equation}
Thus, 
\begin{equation}
\begin{split}
&\mathcal{F}\left[ Q^{+}_{e}(g,f) \right](\xi) \\
=&\int_{\bR^{3}} \int_{\bR^{3}} \int_{\Sd^{2}} g(\vs) f(v) b\left(\frac{v-\vs}{|v-\vs|}\cdot \sigma\right) \e^{ -i\frac{v+\vs}{2}\cdot\xi}\e^{-i \left(\frac{1-e}{4}(v-\vs) + \frac{1+e}{4}|v-\vs|\sigma\right)\cdot\xi} \,\rd \sigma \,\rd \vs \,\rd v \\
=& \int_{\bR^{3}} \int_{\bR^{3}} \int_{\Sd^{2}} g(\vs) f(v) b\left(\frac{\xi}{|\xi|}\cdot \sigma\right) \e^{ -i\frac{v+\vs}{2}\cdot\xi} \e^{-i \left(\frac{1-e}{4}\xi + \frac{1+e}{4}|\xi|\sigma\right)\cdot(v-\vs)} \,\rd \sigma \,\rd \vs \,\rd v\\
=& \int_{\bR^{3}} \int_{\bR^{3}} \int_{\Sd^{2}} g(\vs) f(v) b\left(\frac{\xi}{|\xi|}\cdot \sigma\right) \e^{ -iv\cdot\left( \frac{\xi}{2} + \frac{1-e}{4}\xi + \frac{1+e}{4}|\xi|\sigma \right)} \e^{-i \vs \cdot\left( \frac{\xi}{2} - \frac{1-e}{4}\xi - \frac{1+e}{4}|\xi|\sigma \right) } \,\rd \sigma \,\rd \vs \,\rd v\\
=&  \int_{\Sd^{2}} b\left(\frac{\xi}{|\xi|}\cdot \sigma\right) \hat{f}(\xie^{+}) \hat{g}(\xie^{-})   \,\rd \sigma,
\end{split}
\end{equation}
where, unlike the elastic case, the $ \xi^{+} $ and $ \xi^{-} $ are defined as
\begin{equation}
\xie^{+} = \frac{\xi}{2} + \frac{1-e}{4}\xi + \frac{1+e}{4}|\xi|\sigma, \quad
\xie^{-} = \frac{\xi}{2} - \frac{1-e}{4}\xi - \frac{1+e}{4}|\xi|\sigma.
\end{equation}

\section*{Acknowledgement}
\label{sec:ack}
The author would like to express sincere gratitude to
Prof.~Tong Yang for his giving the related topic and constant support. Also the author would thank Dr.~Shuaikun Wang for his helpful discussion and valuable advice.

\bibliographystyle{plain}
\bibliography{QI_bibtex}

\begin{thebibliography}{10}

\bibitem{Alexreview2009}
R.~Alexandre.
\newblock A review of {B}oltzmann equation with singular kernels.
\newblock {\em Kinet. Relat. Models}, 2(4):551--646, 2009.

\bibitem{AL2014CMP}
R.~Alonso and B.~Lods.
\newblock Boltzmann model for viscoelastic particles: asymptotic behavior,
  pointwise lower bounds and regularity.
\newblock {\em Comm. Math. Phys.}, 331(2):545--591, 2014.

\bibitem{AL2010}
R.~J. Alonso and B.~Lods.
\newblock Free cooling and high-energy tails of granular gases with variable
  restitution coefficient.
\newblock {\em SIAM J. Math. Anal.}, 42(6):2499--2538, 2010.

\bibitem{BLM2015electronicProb}
F.~Bassetti, L.~Ladelli, and D.~Matthes.
\newblock Infinite energy solutions to inelastic homogeneous {B}oltzmann
  equations.
\newblock {\em Electron. J. Probab.}, 20:no. 89, 34, 2015.

\bibitem{BCT2006}
M.~Bisi, J.~A. Carrillo, and G.~Toscani.
\newblock Decay rates in probability metrics towards homogeneous cooling states
  for the inelastic {M}axwell model.
\newblock {\em J. Stat. Phys.}, 124(2-4):625--653, 2006.

\bibitem{Bobylev1976invariant}
A.~V. Bobylev.
\newblock A class of invariant solutions of the {B}oltzmann equation.
\newblock {\em Dokl. Akad. Nauk SSSR}, 231(3):571--574, 1976.

\bibitem{Bobylev2000inelastic}
A.~V. Bobylev, J.~A. Carrillo, and I.~M. Gamba.
\newblock On some properties of kinetic and hydrodynamic equations for
  inelastic interactions.
\newblock {\em J. Statist. Phys.}, 98(3-4):743--773, 2000.

\bibitem{BC2002selfsimilarapplication}
A.~V. Bobylev and C.~Cercignani.
\newblock Self-similar solutions of the {B}oltzmann equation and their
  applications.
\newblock {\em J. Statist. Phys.}, 106(5-6):1039--1071, 2002.

\bibitem{BC2003}
A.~V. Bobylev and C.~Cercignani.
\newblock Self-similar asymptotics for the {B}oltzmann equation with inelastic
  and elastic interactions.
\newblock {\em J. Statist. Phys.}, 110(1-2):333--375, 2003.

\bibitem{BCG2008lecture}
A.~V. Bobylev, C.~Cercignani, and I.~M. Gamba.
\newblock Generalized kinetic {M}axwell type models of granular gases.
\newblock In {\em Mathematical models of granular matter}, volume 1937 of {\em
  Lecture Notes in Math.}, pages 23--57. Springer, Berlin, 2008.

\bibitem{BCG2009selfsimilar}
A.~V. Bobylev, C.~Cercignani, and I.~M. Gamba.
\newblock On the self-similar asymptotics for generalized nonlinear kinetic
  {M}axwell models.
\newblock {\em Comm. Math. Phys.}, 291(3):599--644, 2009.

\bibitem{Brilliantov2004}
N.~V. Brilliantov and T.~P\"{o}schel.
\newblock {\em Kinetic theory of granular gases}.
\newblock Oxford Graduate Texts. Oxford University Press, Oxford, 2004.

\bibitem{cannone2010infinite}
M.~Cannone and G.~Karch.
\newblock Infinite energy solutions to the homogeneous {B}oltzmann equation.
\newblock {\em Comm. Pure Appl. Math.}, 63(6):747--778, 2010.

\bibitem{cannone2013selfsimilar}
M.~Cannone and G.~Karch.
\newblock On self-similar solutions to the homogeneous {B}oltzmann equation.
\newblock {\em Kinet. Relat. Models}, 6(4):801--808, 2013.

\bibitem{Toscani1999}
E.~A. Carlen, E.~Gabetta, and G.~Toscani.
\newblock Propagation of smoothness and the rate of exponential convergence to
  equilibrium for a spatially homogeneous {M}axwellian gas.
\newblock {\em Comm. Math. Phys.}, 199(3):521--546, 1999.

\bibitem{Cercignani}
C.~Cercignani.
\newblock {\em The {B}oltzmann {E}quation and {I}ts {A}pplications}.
\newblock Springer-Verlag, New York, 1988.

\bibitem{Cercignani1995}
C.~Cercignani.
\newblock Recent developments in the mechanics of granular materials.
\newblock In {\em Fisica Matematicae Ingeneria Delle Strutture}, pages
  119--132. Pitagora Editrice, Bologna, 1995.

\bibitem{CMWY2016}
Y.~K. Cho, Y.~Morimoto, S.~Wang, and T.~Yang.
\newblock Probability measures with finite moments and the homogeneous
  {B}oltzmann equation.
\newblock {\em SIAM J. Math. Anal.}, 48(4):2399--2413, 2016.

\bibitem{CH2014}
S.~H. Choi and S.~Y. Ha.
\newblock Global existence of classical solutions to the inelastic
  {V}lasov-{P}oisson-{B}oltzmann system.
\newblock {\em J. Stat. Phys.}, 156(5):948--974, 2014.

\bibitem{DesvillettesNote2003}
L.~Desvillettes.
\newblock About the use of the {F}ourier transform for the {B}oltzmann
  equation.
\newblock volume~2*, pages 1--99. 2003.
\newblock Summer School on ``Methods and Models of Kinetic Theory'' (M\&MKT
  2002).

\bibitem{Toscani1995}
G.~Gabetta, G.~Toscani, and B.~Wennberg.
\newblock Metrics for probability distributions and the trend to equilibrium
  for solutions of the {B}oltzmann equation.
\newblock {\em J. Statist. Phys.}, 81(5-6):901--934, 1995.

\bibitem{GambaDiffusively}
I.~M. Gamba, V.~Panferov, and C.~Villani.
\newblock On the {B}oltzmann equation for diffusively excited granular media.
\newblock {\em Comm. Math. Phys.}, 246(3):503--541, 2004.

\bibitem{GradCutoff}
H.~Grad.
\newblock Asymptotic theory of the {B}oltzmann equation. {II}.
\newblock In {\em Rarefied {G}as {D}ynamics ({P}roc. 3rd {I}nternat. {S}ympos.,
  {P}alais de l'{UNESCO}, {P}aris, 1962), {V}ol. {I}}, pages 26--59. Academic
  Press, New York, 1963.

\bibitem{LuMouhot2012}
X.~Lu and C.~Mouhot.
\newblock On measure solutions of the {B}oltzmann equation, part {I}: moment
  production and stability estimates.
\newblock {\em J. Differential Equations}, 252(4):3305--3363, 2012.

\bibitem{LuMouhot2015}
X.~Lu and C.~Mouhot.
\newblock On measure solutions of the {B}oltzmann equation, {P}art {II}: {R}ate
  of convergence to equilibrium.
\newblock {\em J. Differential Equations}, 258(11):3742--3810, 2015.

\bibitem{MM2006hardsphere2}
S.~Mischler and C.~Mouhot.
\newblock Cooling process for inelastic {B}oltzmann equations for hard spheres.
  {II}. {S}elf-similar solutions and tail behavior.
\newblock {\em J. Stat. Phys.}, 124(2-4):703--746, 2006.

\bibitem{MM2009inelasticlimit}
S.~Mischler and C.~Mouhot.
\newblock Stability, convergence to self-similarity and elastic limit for the
  {B}oltzmann equation for inelastic hard spheres.
\newblock {\em Comm. Math. Phys.}, 288(2):431--502, 2009.

\bibitem{MM2006hardsphere1}
S.~Mischler, C.~Mouhot, and M.~Rodriguez~Ricard.
\newblock Cooling process for inelastic {B}oltzmann equations for hard spheres.
  {I}. {T}he {C}auchy problem.
\newblock {\em J. Stat. Phys.}, 124(2-4):655--702, 2006.

\bibitem{morimoto2012remark}
Y.~Morimoto.
\newblock A remark on {C}annone-{K}arch solutions to the homogeneous
  {B}oltzmann equation for {M}axwellian molecules.
\newblock {\em Kinet. Relat. Models}, 5(3):551--561, 2012.

\bibitem{MWY2015}
Y.~Morimoto, S.~Wang, and T.~Yang.
\newblock A new characterization and global regularity of infinite energy
  solutions to the homogeneous {B}oltzmann equation.
\newblock {\em J. Math. Pures Appl. (9)}, 103(3):809--829, 2015.

\bibitem{morimoto2016measure}
Y.~Morimoto, S.~Wang, and T.~Yang.
\newblock Measure valued solutions to the spatially homogeneous {B}oltzmann
  equation without angular cutoff.
\newblock {\em J. Stat. Phys.}, 165(5):866--906, 2016.

\bibitem{MYZ2017convergence}
Y.~Morimoto, T.~Yang, and H.~Zhao.
\newblock Convergence to self-similar solutions for the homogeneous {B}oltzmann
  equation.
\newblock {\em J. Eur. Math. Soc. (JEMS)}, 19(8):2241--2267, 2017.

\bibitem{ToscaniVillani1999}
G.~Toscani and C.~Villani.
\newblock Probability metrics and uniqueness of the solution to the {B}oltzmann
  equation for a {M}axwell gas.
\newblock {\em J. Statist. Phys.}, 94(3-4):619--637, 1999.

\bibitem{Villani02}
C.~Villani.
\newblock A review of mathematical topics in collisional kinetic theory.
\newblock In S.~Friedlander and D.~Serre, editors, {\em Handbook of
  Mathematical Fluid Mechanics}, volume~I, pages 71--305. North-Holland, 2002.

\bibitem{Villani2006granularmaterials}
C.~Villani.
\newblock Mathematics of granular materials.
\newblock {\em J. Stat. Phys.}, 124(2-4):781--822, 2006.

\end{thebibliography}

\end{document}